\documentclass{article}

\usepackage[preprint]{corl_2021}
\usepackage{times}

\usepackage{multicol}
\usepackage{booktabs}
\usepackage{multirow}
\usepackage{setspace}
\usepackage{tabularx}
\usepackage{xfrac}
\usepackage{mathtools}
\usepackage{amsthm}
\usepackage{subcaption}
\usepackage{adjustbox}
\usepackage{multicol}
\usepackage{wrapfig}

\usepackage{mymacros}

\title{Safe Nonlinear Control Using Robust Neural Lyapunov-Barrier Functions}

\author{
  Charles Dawson$^1$, Zengyi Qin$^1$, Sicun Gao$^2$, Chuchu Fan$^1$ \\
  $^1$ Massachusetts Institute of Technology, \texttt{\{cbd, qinzy, chuchu\}@mit.edu}\\
  $^2$ University of California, San Diego, \texttt{sicung@ucsd.edu}
}

\newcommand{\xg}{x_{\mathrm{goal}}}
\newcommand{\cX}{\mathcal{X}}
\newcommand{\xu}{\mathcal{X}_{\mathrm{unsafe}}}
\newcommand{\xs}{\mathcal{X}_{\mathrm{safe}}}

\mathchardef\mhyphen="2D 

\newtheorem{definition}{Definition}
\newtheorem{lemma}{Lemma}
\newtheorem{theorem}{Theorem}

\begin{document}

\maketitle

\begin{abstract}
Safety and stability are common requirements for robotic control systems; however, designing safe, stable controllers remains difficult for nonlinear and uncertain models. We develop a model-based learning approach to synthesize robust feedback controllers with safety and stability guarantees. We take inspiration from robust convex optimization and Lyapunov theory to define robust control Lyapunov barrier functions that generalize despite model uncertainty. We demonstrate our approach in simulation on problems including car trajectory tracking, nonlinear control with obstacle avoidance, satellite rendezvous with safety constraints, and flight control with a learned ground effect model. Simulation results show that our approach yields controllers that match or exceed the capabilities of robust MPC while reducing computational costs by an order of magnitude.
\end{abstract}

\section{Introduction}
\begin{wrapfigure}{r}{0.35\linewidth}
    \centering
    \vspace{-2em}
    \includegraphics[width=45mm]{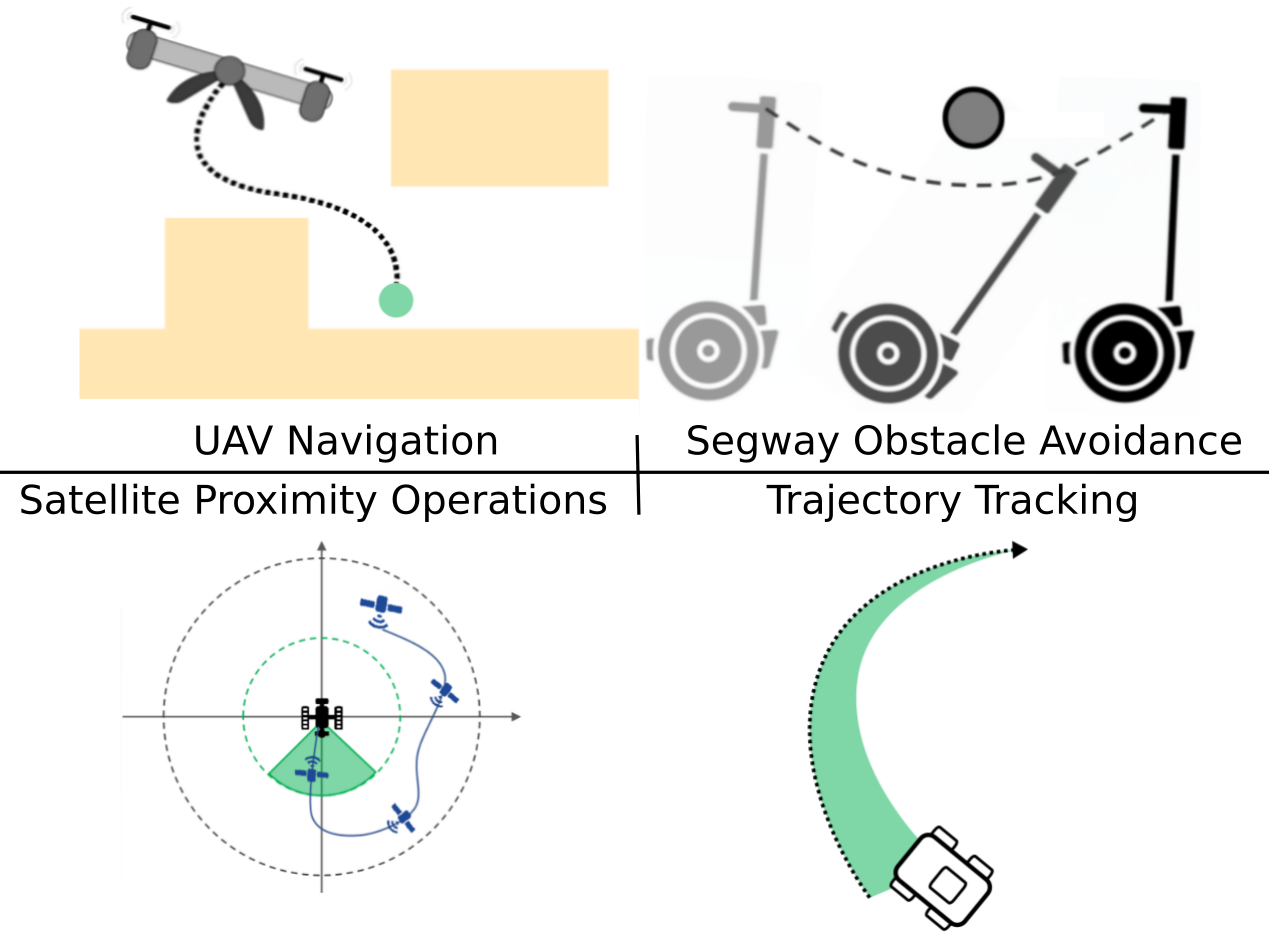}
    \caption{Safe control problems considered in Section~\ref{experiments}.}
    \label{fig:headline}
    \vspace{-1em}
\end{wrapfigure}
Robot control systems are challenging to design, not least because of the problems of \textit{task complexity} and \textit{model uncertainty}. Robotics control problems like those in Fig.~\ref{fig:headline} often involve both safety and stability requirements, where the controller must drive the system towards a goal state while avoiding unsafe regions. Complicating matters, the model used to design the controller is seldom a perfect representation of the physical plant, and so controllers must account for uncertainty in any parameters (e.g. mass, friction, or unmodeled effects) that vary between the engineering model and true plant. Automatically synthesizing safe, stable, and robust controllers for nonlinear reach-avoid tasks is a long-standing open problem in controls. In this paper, we address this problem with a novel approach to robust model-based learning. Our work presents a unified framework for handling both model uncertainty and complex safety and stability specifications.

Over the years, several approaches have been proposed to solve this problem. In one view, reach-avoid can be treated as an optimal control problem and solved using model predictive control (MPC) schemes and their robust variants. Robust MPC promises a method for general-purpose controller synthesis, finding an optimal control signal given only a model of the system and a specification of the task. However, there are a number of recognized disadvantages of robust MPC. First, there are currently no techniques for guaranteeing the safety, stability, or recursive feasibility of robust MPC beyond the linear case \cite{rmpc}. Second, many sources of model uncertainty (e.g. mass or friction) are multiplicative in the dynamics, but robust MPC algorithms are typically limited to additive uncertainty \cite{rmpc,Lofberg2003}. Finally, MPC is computationally expensive, making it difficult to achieve high control frequencies in practice \cite{Levine2019}.

An alternative method for synthesizing safe, stable controllers comes from Lyapunov theory, through the use of control Lyapunov and control barrier functions (resp., CLFs and CBFs, \cite{Ames2017a}) --- certificates that prove the stability and safety of a control system, respectively. CLFs and CBFs are similar to standard Lyapunov and barrier functions, but they can be used to synthesize a controller rather than just verifying the performance of a closed-loop system. Unfortunately, CLF and CBF certificates are very difficult construct in general, particularly for systems with nonlinear dynamics \cite{Giesl2015}.

The most recent set of methods promising general-purpose controller synthesis come from the field of learning for control; for instance, using reinforcement learning \cite{Cheng2019,Han2020} or supervised learning \cite{Chang2019,Sun2020,Qin2021,Tsukamoto2020}. However, the introduction of learning-enabled components into safety-critical control tasks raises questions about soundness, robustness, and generalization. Some learning-based control techniques incorporate certificates such as Lyapunov functions \cite{Chang2019}, barrier functions \cite{Dean2020a,Qin2021,Peruffo2020}, and contraction metrics \cite{Sun2020,Tsukamoto2020} to prove the soundness of learned controllers. Unfortunately, these certificates' guarantees are sensitive to uncertainties in the underlying model. In particular, if the model used during training differs from that encountered during deployment, then guarantees on safety and stability may no longer hold.

Our main contribution is a learning-based framework for synthesizing robust nonlinear feedback controllers from safety and stability specifications. This contribution has two parts. First, we provide a novel extension of control Lyapunov barrier functions to robust control, defining a robust control Lyapunov barrier function (robust CLBF). Second, we develop a model-based approach to learning robust CLBFs, which we use to derive a safe controller using techniques from robust convex optimization. Other methods for learning Lyapunov and barrier certificates exist, but a key advantage of our approach is that we learn certificates with explicit robustness guarantees, enabling generalization beyond the system parameters seen during training. We demonstrate our approach on a range of challenging control problems, including trajectory tracking, nonlinear control with obstacle avoidance, flight control with a learned model of ground effect, and a satellite rendezvous problem with non-convex safety constraints, comparing our approach with robust MPC. In all of these experiments, we find that our method either matches or exceeds the performance of robust MPC while reducing computational cost at runtime by at least a factor of 10.

\section{Related Work}

This work builds on a rich history of certificate-based control theory, including classical Lyapunov functions as well as more recent approaches such as control Lyapunov functions (CLFs \cite{Artstein1983,Ames2014}) and control barrier functions (CBFs \cite{Ames2019}, a generalization of artificial potential fields~\cite{Singletary2020_cbf_apf}). The majority of classical certificate-based controllers rely on hand-designed certificates \cite{Choi2020,Castaneda2020}, but these can be difficult to obtain for nonlinear or high-dimensional systems. Some automated techniques exist for synthesizing CLFs and CBFs; however, many of these techniques (such as finding a Lyapunov function as the solution of a partial differential equation) are computationally intractable for many practical applications \cite{Giesl2015}. Other automated synthesis techniques are based on convex optimization, particularly sum-of-squares programming (SOS, \cite{Ahmadi2016}), but are limited to systems with polynomial dynamics and do not scale favorably with the dimension of the system.

A promising line of work in this area is to use neural networks to learn certificate functions. These techniques range in complexity from verifying the stability of a given control system \cite{Abate2020,Richards2018} to simultaneously learning a control policy and certificate \cite{Sun2020,Chang2019,Qin2021}. Most of these works do not explicitly consider robustness to model uncertainty, although contraction metrics may be used to certify robustness to bounded additive disturbance \cite{Sun2020}.

Most approaches to handling model uncertainty in the context of certificate-guided learning for control involve online adaptation. For example, \cite{Choi2020,Taylor2019} assume that a CLF or CBF are given and learn the unmodeled residuals in the CLF and CBF derivatives. When combined with a QP-based CLF/CBF controller, this technique enables adaptation to model uncertainty but relies on a potentially unsafe exploration phase. Although safe adaptation strategies exist, the main drawback with these techniques is their reliance on a hand-designed CLF and CBF, which are non-trivial to synthesize for nonlinear systems. Additionally, combined CLF/CBF controllers are prone to getting stuck when the feasible sets of the CLF and CBF no longer intersect.

Online optimization-based control techniques such as model-predictive control (MPC) are also relevant as a general-purpose control synthesis strategy. However, the computational complexity of MPC, and particularly robust MPC, is a widely-recognized issue, particularly when considering deployment to resource-constrained robotic systems such as UAVs \cite{rmpc,Levine2019}. We revisit the computational cost of robust MPC, particularly as compared with the cost of our proposed method, in Section~\ref{experiments}. Some approaches apply learning to characterize uncertainty in system dynamics and augment a robust MPC scheme \cite{Fan2020}, but these methods do not fundamentally change the computational burden of MPC. Other methods rely on imitation learning to recreate an MPC-based policy online \cite{Kahn2016}, but these methods can encounter difficulties in generalizing beyond the training dataset.

A number of techniques from classical nonlinear control also deserve mention, such as sliding mode and adaptive controllers. These methods do not directly support state constraints and so must be paired with a separate trajectory planning layer \cite{slotine_li_1991}. Another drawback is that these techniques require significant effort to manually derive appropriate feedback control laws, and we are primarily interested in automated techniques for controller synthesis.

\section{Preliminaries and Background}\label{prelim}

We consider continuous-time, control-affine dynamical systems of the form $\dot{x} = f_\theta(x) + g_\theta(x)u$, where $x \in \cX \subseteq \R^n$, $u \in \R^\ell$, and $f_\theta: \R^n \rightarrow{\R}^n$ and $g_\theta: \R^n \rightarrow \R^{n \times \ell}$ are smooth functions modeling control-affine nonlinear dynamics. We assume that $f_\theta$ and $g_\theta$ depend on model parameters $\theta \in \Theta \subseteq \R^r$ and are affine in those parameters for any fixed $x$. This assumption on the dynamics is not restrictive; it covers many physical systems with uncertainty in inertia, damping, or friction (e.g. rigid-body dynamics or systems described by the manipulator equations), and it includes bounded additive and multiplicative disturbance as a special case. We also assume that $f_\theta$ and $g_\theta$ are Lipschitz but make no further assumptions, allowing us to consider cases when components of  $f_\theta$ and $g_\theta$ are learned from experimental data. For concision, we will use $f$ and $g$ (without subscript) to refer to the dynamics evaluated with nominal parameters $\theta_0 \in \Theta$. In this paper, we consider the following control synthesis problem:
\begin{definition}[Robust Safe Control Problem]\label{reach_avoid}
Given a control-affine system with uncertain parameters $\theta \in \Theta$, a goal configuration $\xg$, a set of unsafe states $\xu \subseteq \cX$, and a set of safe states $\xs \subseteq \cX$ (such that $\xs \cap \xu = \emptyset$ and $\xg \in \xs$), find a control policy $u  = \pi(x)$ such that all trajectories $x(t)$ satisfying $\dot{x} = f_\theta(x) + g_\theta(x)\pi(x)$ and $x(0) \in \xs$ have the following properties for any parameters $\theta$:
\begin{minipage}{1\textwidth}
    \begin{itemize}
        \setlength{\itemsep}{20pt}
        \setlength{\parskip}{20pt}
        \begin{multicols}{2}
            \item[] \textbf{Reachability} of $\xg$ with tolerance\\$\delta$: $\lim_{t\to\infty}\norm{x(t) - \xg} \leq \delta$
            \item[] \textbf{Safety}: $x(t_1) \in \xs$ implies $x(t_2) \notin \xu \ \forall \ t_2 \geq t_1$
        \end{multicols}
    \end{itemize}
\end{minipage}

\end{definition}

Simply put, we wish to \textit{reach} the goal $\xg$ while \textit{avoiding} the unsafe states $\xu$. We use the notion of reachability instead of asymptotic stability to permit (small) steady-state error; in the following we will use ``stable'' as shorthand for reachability. Note that we do not require $\xs \cup \xu = \cX$, as it will be made clear in the following discussion that we need a non-empty boundary layer $\cX \setminus (\xs \cup \xu)$ to allow for flexibility in finding a safety certificate.

Lyapunov theory provides tools that are naturally suited to reach-avoid problems: control Lyapunov functions (for stability) and control barrier functions (for safety \cite{Ames2017a}). To avoid issues arising from learning two separate certificates, we rely on a single, unifying certificate known as a control Lyapunov barrier function (CLBF). Our definition of CLBFs is related to those in \cite{Romdlony2016} and \cite{Xiao2021} (differing from the formulation in \cite{Romdlony2016} by a constant offset $c$, and differing from~\cite{Xiao2021} where safety and reachability are proven using two separate CLBFs). We begin by providing a standard definition of a CLBF in the non-robust case, but in the next section we provide a novel, robust extension of CLBF theory before demonstrating how neural networks may be used to synthesize these functions for a general class of dynamical system. In the following, we denote $L_f V$ as the Lie derivative of $V$ along $f$.

\begin{definition}[CLBF]
A function $V: \cX \to \R$ is a CLBF if, for some $c,\lambda > 0$, 
\begin{subequations}
    \vspace{-2em}
    \begin{multicols}{2}
    \begin{align}
        & V(\xg) = 0 \label{eq:clbf_goal_cond}\\
        & V(x) > 0\ \forall\ x \in \cX \setminus \xg \label{eq:clbf_pos_cond}
    \end{align}
    \begin{align}
        &\nonumber\\
        & V(x) \leq c\ \forall\ x \in \xs \label{eq:clbf_safe_cond} \\
        & V(x) > c\ \forall\ x \in \xu \label{eq:clbf_unsafe_cond}
    \end{align}
    \end{multicols}
    \vspace{-2.5em}
    \begin{align}
        \inf_u L_f V + L_g V u + \lambda V(x) \leq 0 \ \forall \ x \in \cX \setminus \xg \label{eq:clbf_decrease_cond}
    \end{align}
\end{subequations}
\end{definition}
\vspace{-1em}
Intuitively, we can think of a CLBF as a special case of a control Lyapunov function where the safe and unsafe regions are contained in sub- and super-level sets, respectively. If we define a set of admissible controls $K(x) = \set{u\ |\ L_f V + L_g V u + \lambda V \leq 0}$, then we arrive at a theorem proving the stability and safety of any controller that outputs elements of this set (the proof is included in the supplementary material).
\begin{theorem}\label{clbf_both}
    If $V(x)$ is a CLBF then any control policy $\pi(x) \in K(x) \ \forall\ x \in \cX$ will be both safe and stable, in the sense of Definition~\ref{reach_avoid}.
\end{theorem}


Based on these results, we can define a CLBF-based controller, analogous to the CLF/CBF-based controller in \cite{Choi2020} but without the risk of conflicts between the CLF and CBF conditions, relying on the CLBF $V$ and some nominal controller $\pi_{\mathrm{nominal}}$ (e.g. the LQR policy):
\begin{align}
    \pi_{\rm{CLBF}}(x) = \argmin_{u} \quad & \frac{1}{2} \norm{u - \pi_{\mathrm{nominal}}(x)}^2 \tag{CLBF-QP}\label{eq:clbf_qp}\\
    \text{s.t.} \quad & L_f V + L_g V u + \lambda V \leq 0
\end{align}
It should be clear that $\pi_{\textrm{CLBF}}(x) \in K(x) \ \forall\ x \in \cX \setminus \xg$, so this controller will result in a system that is certifiably safe and stable (with the CLBF $V$ acting as the certificate). The nominal control signal $\pi_{\mathrm{nominal}}$ is included to encourage smoothness in the solution $\pi_{\rm{CLBF}}(x)$, particularly near the desired fixed point at $\xg$ where $\dot{V}$ becomes small. CLBFs provide a single, unified certificate of safety and stability; however, some significant issues remain. In particular, how do we guarantee that a CLBF will generalize beyond the nominal parameters?

\section{Robust CLBF Certificates for Safe Control}

In this section, we extend the definition of CLBFs to provide explicit robustness guarantees, and we present a key theorem proving the soundness of robust CLBF-based control.


\begin{definition}[Robust CLBF, rCLBF]\label{rclbf_def}
A function $V: \cX \to \R$ is a robust CLBF for bounded parametric uncertainty $\theta \in \Theta$, where $\Theta$ is the convex hull of scenarios $\theta_1, \theta_2, \ldots, \theta_{n_s}$ if the standard CLBF conditions \eqref{eq:clbf_goal_cond}--\eqref{eq:clbf_unsafe_cond} hold, the dynamics $f$ and $g$ are affine with respect to $\theta$, and $\forall \ x \in \cX \setminus \xg$ there exist $c,\lambda > 0$ such that
    \begin{align}
        & \inf_u L_{f_{\theta_i}} V + L_{g_{\theta_i}} V u + \lambda V(x) \leq 0 \qquad \forall i = 1,\ldots,n_s \label{eq:rclbf_decrease_cond}
    \end{align}
\end{definition}
As in the non-robust case, we define the set of admissible controls for a robust CLBF, $K_r(x) = \set{u\ |\ L_{f_{\theta_i}} V + L_{g_{\theta_i}} V u + \lambda V \leq 0 \ \forall\ i=0,\ldots,n_s}$, and the corresponding QP-based controller, the soundness of which is given by Theorem~\ref{rclbf}:
\begin{align}
    \pi_{\mathrm{rCLBF}} = &\argmin_{u}\ \norm{u - \pi_{\mathrm{nominal}}}^2 \tag{rCLBF-QP} \label{eq:rclbf_qp}\\
    \text{s.t.}&\ L_{f_{\theta_i}} V + L_{g_{\theta_i}} V u + \lambda V \leq 0;\ i=0, \ldots, n_s \label{eq:rclbf_qp_cons}
\end{align}
\begin{theorem}\label{rclbf}
    If $V(x)$ is a robust CLBF, then any control policy $\pi(x) \in K_r(x) \ \forall\ x \in \cX$ will be both safe and stable, in the sense of Definition~\ref{reach_avoid}, when executed on a system $f_\theta$, $g_\theta$ with uncertain parameters $\theta \in \Theta$ (where $\Theta$ is the convex hull of scenarios $\theta_0,\ldots,\theta_{n_s}$).
\end{theorem}
\begin{proof}
    See the supplementary materials.
\end{proof}
\vspace{-1em}
This result demonstrates the soundness and robustness of an rCLBF-based controller, but does not provide a means to construct a valid rCLBF. In the next section, we will present an automated model-based learning approach to rCLBF synthesis, yielding a general framework for solving robust safe control problems even for systems with complex, nonlinear, or partially-learned dynamics.

\section{Learning Robust CLBFs}\label{learning_approach}

A persistent challenge in using of certificate-based controllers is the difficulty of finding valid certificates, especially for systems with nonlinear dynamics and complex specifications of $\xs$ and $\xu$ (e.g. obstacle avoidance). Taking inspiration from recent advances in certificate-guided learning for control \cite{Chang2019, Qin2021}, we employ a model-based supervised learning framework to synthesize an rCLBF-based controller. 
The controller architecture is comprised of three main parts: the rCLBF $V$, a proof controller $\pi_{\mathrm{NN}}$, and the QP-based controller \eqref{eq:rclbf_qp}. We parameterize $V: \cX \to \R$ and $\pi_{\mathrm{NN}}: \cX \to \R^\ell$ as neural networks. These networks are trained offline, where $\pi_{\rm{NN}}$ is used to prove that the feasible set of \eqref{eq:rclbf_qp} is non-empty, then $V$ is evaluated online to provide the parameters of \eqref{eq:rclbf_qp}, which is solved to find the control input.
%
%
In the offline training stage, our primary goal is finding an rCLBF $V(x)$ such that the conditions of Definition~\ref{rclbf_def} are satisfied. To ensure \eqref{eq:clbf_pos_cond}, we define $V(x) = w^T(x)w(x) \geq 0$, where $w$ is the activation vector of the last hidden layer of the $V$ neural network. To train $V$ such that conditions \eqref{eq:clbf_goal_cond}, \eqref{eq:clbf_safe_cond}, \eqref{eq:clbf_unsafe_cond}, and \eqref{eq:rclbf_decrease_cond} are satisfied over the domain of interest, we sample $N_{\mathrm{train}}$ points uniformly at random from $\cX$ to yield a population of training points $x$, then define the empirical loss:
\begin{align}
    \mathcal{L}_{\mathrm{rCLBF}} &= V(\xg)^2 + a_1 \frac{1}{N_{\mathrm{safe}}} \sum_{x \in \xs} \left[\epsilon + V(x) - c \right]_+ + a_2 \frac{1}{N_{\mathrm{unsafe}}} \sum_{x \in \xu} \left[\epsilon + c - V(x)\right]_+ \nonumber\\
    &\quad\quad + \frac{a_3}{n_s N_{train}} \sum_{x} r(x) \sum_{i=0}^{n_s} [\epsilon + L_{f_{\theta_i}} V(x) + L_{g_{\theta_i}} V(x) \pi_{\mathrm{NN}}(x) + \lambda V(x)]_+ 
\end{align}
where $a_1$--$a_3$ are positive tuning parameters, $\epsilon > 0$ is a small parameter (typically $0.01$) that allows us to encourage strict inequality satisfaction and enables generalization claims, $N_{\mathrm{safe}}$ and $N_{\mathrm{unsafe}}$ are the number of points in the training sample in $\xs$ and $\xu$, respectively, and $[\circ]_+ = \max(\circ, 0)$ is the ReLU function. The terms in this empirical loss are directly linked to conditions \eqref{eq:clbf_goal_cond}, \eqref{eq:clbf_safe_cond}, \eqref{eq:clbf_unsafe_cond}, and \eqref{eq:rclbf_decrease_cond} such that each term is zero if the corresponding condition is satisfied at all $N_{\mathrm{train}}$ training points. For example, the final term in this loss is designed to encourage satisfaction of the robust CLBF decrease condition \eqref{eq:rclbf_decrease_cond}. The factor $r(x)$ in the final term is computed by solving~\eqref{eq:rclbf_qp} at each training point and computing the maximum violation of constraint~\eqref{eq:rclbf_qp_cons}, such that $r(x) = 0$ when the QP has a feasible solution and $r(x) > 0$ otherwise. This loss is optimized using stochastic gradient descent, alternating epochs between training the $V$ and $\pi_{\mathrm{NN}}$ networks. During training, we rely on $\pi_{\mathrm{NN}}$ to compute the time derivative of $V(x)$ in the final term of the loss. To provide a training signal for $\pi_{\mathrm{NN}}$, we define an additional loss $\mathcal{L}_\pi = \norm{\pi_{\mathrm{NN}} - \pi_{\mathrm{nominal}}}^2$, where $\pi_{\mathrm{nominal}}$ is a nominal controller (e.g. a policy derived from an LQR approximation). The parameters of $V$ and $\pi_{\mathrm{NN}}$ are optimized using the combined loss $\mathcal{L} = \mathcal{L}_{\mathrm{rCLBF}} + (10^{-5})\mathcal{L}_\pi$. The small weight applied to $\mathcal{L}_\pi$ ensures that the training process prioritizes satisfying the CLBF conditions.

An important detail of our control architecture is that the learned control policy $\pi_{\mathrm{NN}}$ is used primarily to demonstrate that the feasible set of \eqref{eq:rclbf_qp} is non-empty. We are not required to use $\pi_{\mathrm{NN}}$ at execution time; we can choose any control policy from the admissible set $K_{r}(x)$. In the online stage, we rely on an optimization-based controller \eqref{eq:rclbf_qp}, which solves a small quadratic program with $n_s$ constraints and $\ell$ variables (one for each element of $u$). To ensure that this QP is feasible at execution, we permit a relaxation of the CLBF constraints \eqref{eq:rclbf_qp_cons} and penalize relaxation with a large coefficient in the objective. Once trained, $V$ can be verified using neural-network verification tools \cite{Liu2021}, sampling \cite{Bobiti2018}, or a generalization error bound \cite{Qin2021}. More details on data collection, training, implementation, and verification strategies are included in the supplementary materials.

It is important to note that this training strategy encourages satisfying \eqref{eq:rclbf_decrease_cond} only on the finite set of training points sampled uniformly from the state space; there is no learning mechanism that enforces dense satisfaction of \eqref{eq:rclbf_decrease_cond}. In the supplementary materials, we include plots of 2D sections of the state space showing that \eqref{eq:rclbf_decrease_cond} is satisfied at the majority of points, but there is a relatively small violation on a sparse subset of the state space. Because these violation regions are sparse, the theory of \textit{almost Lyapunov functions} applies \cite{liu2020almost}: small violation regions may induce temporary overshoots (requiring shrinking the certified invariant set), but they do not invalidate the safety and stability assurances of the certificate. Strong empirical results on controller performance in Section~\ref{experiments} support this conclusion, though we admit that good empirical performance is not a substitute for guarantees based on rigorous verification, which we hope to revisit in future work.


\section{Experiments}\label{experiments}

To evaluate the performance of our learned rCLBF-QP controller, we compare against min-max robust model predictive control (as described in \cite{Lofberg2003,lofberg2012}) on a series of simulated benchmark problems representing safe control problems with increasing complexity. The first two concern trajectory tracking, where we wish to limit the tracking error despite uncertainty in the reference trajectory. The next two benchmarks are UAV stabilization problems that add additional safety constraints and increasingly nonlinear dynamics. The last three benchmarks involve highly non-convex safety constraints. The first four benchmarks provide a solid basis for comparison between our proposed method and robust MPC, while the last three demonstrate the power of our approach to generalize to maintain safety even in complex environments.

In each experiment, we vary model parameters randomly in $\Theta$, simulate the performance of the controller, and compute the rate of safety constraints violations and average error relative to the goal $\norm{x-\xg}$ across simulations. These data are reported along with average evaluation time for each controller in Table~\ref{tab:safety_results}. To examine the effect of control frequency on MPC performance, we include results for two different control periods $dt$ for all robust MPC experiments (we also report the horizon length $N$). In some cases we observed that the evaluation time for MPC exceeds the control period; in practice this would lead to the controller failing, but in our experiments we simply ran the simulation slower than real-time. Our robust MPC comparison supports only linear models with bounded additive disturbance; we linearize the systems about the goal point and select an additive disturbance to approximate the disturbance from uncertain model parameters. The following sections will present results from each benchmark separately, and more details are provided in the supplementary materials, including the dynamics and constraints used for each benchmark, as well as the hardware used for training and execution.

\vspace{-1em}
\begin{table}[h]
\caption{Comparison of controller performance under parameter variation}\label{tab:safety_results}
\begin{center}
\begin{spacing}{1.1}
\scalebox{0.8}{
\begin{tabular}{c| c| c| c| c}
\hline
Task & Algorithm & Safety rate & $\norm{x-\xg}$ & Evaluation time (ms) \\
\hline \hline
{Car trajectory tracking}$^1$  & rCLBF-QP & \multicolumn{2}{c|}{\textbf{0.7523}} & \textbf{10.4} \\
Kinematic model & Robust MPC ($dt=\SI{0.1}{s}$, $N=6$) &  \multicolumn{2}{c|}{1.5148} & 194.6 \\
$(n=5,\ell=2,n_s=2)$ & Robust MPC ($dt=\SI{0.25}{s}$, $N=6$) &  \multicolumn{2}{c|}{12.4438} & 172.8 \\
\hline
{Car trajectory tracking}$^1$  & rCLBF-QP & \multicolumn{2}{c|}{1.0340} & \textbf{9.6} \\
Sideslip model & Robust MPC ($dt=\SI{0.1}{s}$, $N=5$) &  \multicolumn{2}{c|}{\textbf{0.1560}} & 336.5 \\
$(n=7,\ell=2,n_s=2)$ & Robust MPC ($dt=\SI{0.25}{s}$, $N=5$) &  \multicolumn{2}{c|}{18.1939} & 316.9 \\
\hline
{3D Quadrotor} & rCLBF-QP & 100\% & 0.4647 & \textbf{9.7} \\
$(n=9,\ell=4,n_s=2)$ & Robust MPC ($dt=\SI{0.10}{s}$, $N=5$) &  100\% & \textbf{0.0980} & 316.2\\
& Robust MPC ($dt=\SI{0.25}{s}$, $N=5$) &  100\% & 63.6303 & 291.0 \\
\hline
{Neural Lander} & rCLBF-QP & 100\% & \textbf{0.1332} & \textbf{13.1} \\
$(n=6,\ell=3,n_s=1)$ & Robust MPC ($dt=\SI{0.10}{s}$, $N=5$) &  100\% & 0.2086 & 247.2 \\
& Robust MPC ($dt=\SI{0.25}{s}$, $N=5$) &  100\% & 0.3267 & 253.2 \\
\hline
{Segway}  & rCLBF-QP & \textbf{100\%} & \textbf{0.0447} & \textbf{4.4} \\
$(n=4,\ell=1,n_s=4)$ & Robust MPC ($dt = \SI{0.10}{s}$, $N=5$) &  21\% & 1.3977 & 214.8 \\
& Robust MPC ($dt = \SI{0.25}{s}$, $N=5$) &  11\% & 1.9725 & 239.1 \\
\hline
2D Quadrotor$^2$  & rCLBF-QP & \multicolumn{2}{c|}{\textbf{83\%}} & \textbf{18.6} \\
$(n=6,\ell=2,n_s=4)$ & Robust MPC ($dt = \SI{0.10}{s}$, $N=5$) &  \multicolumn{2}{c|}{53\%} & 276.9 \\
& Robust MPC ($dt = \SI{0.25}{s}$, $N=5$) &  \multicolumn{2}{c|}{0\%} & 265.2 \\
\hline
{Satellite Rendezvous} & rCLBF-QP & \textbf{100\%} & \textbf{0.1369} & \textbf{8.2} \\
$(n=4,\ell=2)$ & Robust MPC ($dt = \SI{0.10}{s}$, $N=5$) &  39\% & 6.3751 & 187.3 \\
& Robust MPC ($dt = \SI{0.25}{s}$, $N=5$) &  15\% & 9.0592 & 197.4 \\
\hline
\multicolumn{5}{l}{\small $^1$ For car trajectory tracking, we compute maximum tracking error over the trajectory.}\\
\multicolumn{5}{l}{\small $^2$ For 2D quadrotor, we compute \% of trials reaching the goal with tolerance $\delta=0.3$ without collision.} \\
\multicolumn{5}{l}{\small Note: We also implemented SOS optimization to search for a CLBF and controller, but bilinear optimization (as in \cite{Majumdar2013})} \\
\multicolumn{5}{l}{\small did not converge with maximum polynomial degree 10 and a Taylor expansion of the nonlinear dynamics.}
\end{tabular}
}
\end{spacing}
\end{center}
\end{table}
\vspace{-3em}

\subsection{Car trajectory tracking}
\label{sec:tracking}
First, we consider the problem of tracking an \textit{a priori} unknown trajectory using two different car models. In the first model (the kinematic model), the vehicle state is $[x_e, y_e, \delta, v_e, \psi_e]$, representing error relative to the reference trajectory ($\delta$ is the steering angle). The second model (the sideslip model) has state $[x_e, y_e, \delta, v_e, \psi_e, \dot{\psi}_e, \beta]$, where $\beta$ is the sideslip angle \cite{Althoff2017a}. Both models have control inputs for the rate of change of $\delta$ and $v_e$. We assume that the reference trajectory is parameterized by an uncertain curvature: at any point the angular velocity of the reference point can vary on $[-1.5, 1.5]$. The goal point is zero error relative to the reference, and the safety constraint requires maintaining bounded tracking error.

The performance of our controller is shown in Fig.~\ref{fig:car_tracking}. We see that for both models, both our controller and robust MPC are able to track the reference trajectory. However, robust MPC was only successful when run at slower than real-time speeds (with a control period $dt=\SI{0.1}{s}$ roughly twice as fast as the average evaluation time). MPC became unstable when run at a slower control frequency $dt=\SI{0.25}{s}$. In contrast, our rCLBF-QP controller runs in real-time with a control period of $\approx\SI{10}{ms}$ on a laptop computer. This significant improvement in speed is due primarily to the reduction in the size of \eqref{eq:rclbf_qp} relative to that of the QPs used by robust MPC. For example, for the sideslip model, our controller solves a QP with 2 variables and 2 constraints, whereas the robust MPC controller solves a QP with 35 variables and 23 constraints (after pre-compiling using YALMIP \cite{lofberg2012}). Because the learned rCLBF encodes long-term safety and stability constraints into local constraints on the rCLBF derivative, the rCLBF controller requires only a single-step horizon (as opposed to the receding horizon used by MPC).

By comparing performance between these two models, we can discern an important feature of our approach. Increasing the state dimension when moving between models does not substantially increase the evaluation time for our controller (as it does for robust MPC), but it does degrade the tracking performance, suggesting that the number of samples required to train the CLBF to any given level of performance increases with the size of the state space. These examples also highlight a potential drawback of our approach, which relies on a \textit{parameter-invariant} robust CLBF. Because it attempts to find a common rCLBF for all possible parameter values, our controller exhibits some small steady-state error near the goal. This occurs because there is no single control input that renders the goal a fixed point for all possible parameter values and motivates our use of a goal-reaching tolerance in Definition~\ref{reach_avoid}.

\begin{figure}[bh!]
     \hspace*{-1.8cm}
     \begin{minipage}[c][3cm][t]{1.2\textwidth}
       \begin{subfigure}[b]{\textwidth}
         \centering
         \includegraphics[height=2.5cm]{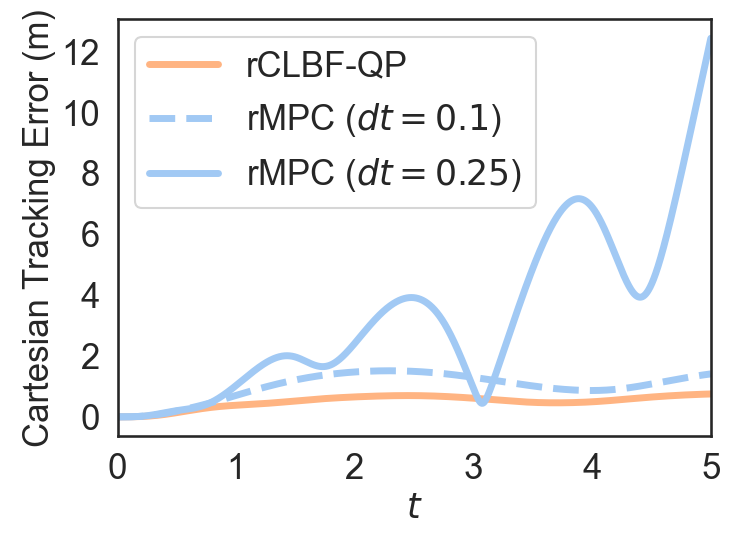}
         \includegraphics[height=2.5cm]{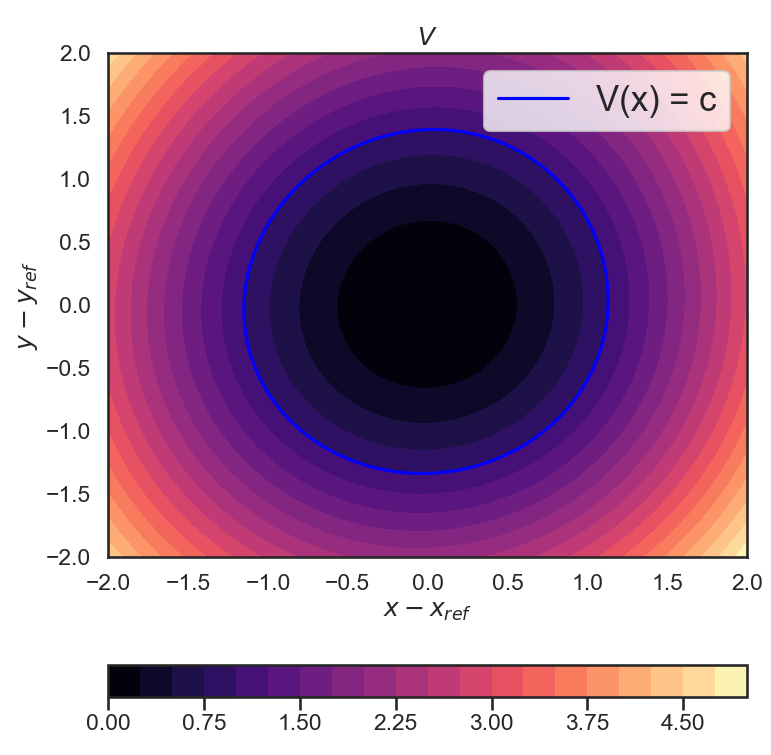}
         \includegraphics[height=2.5cm]{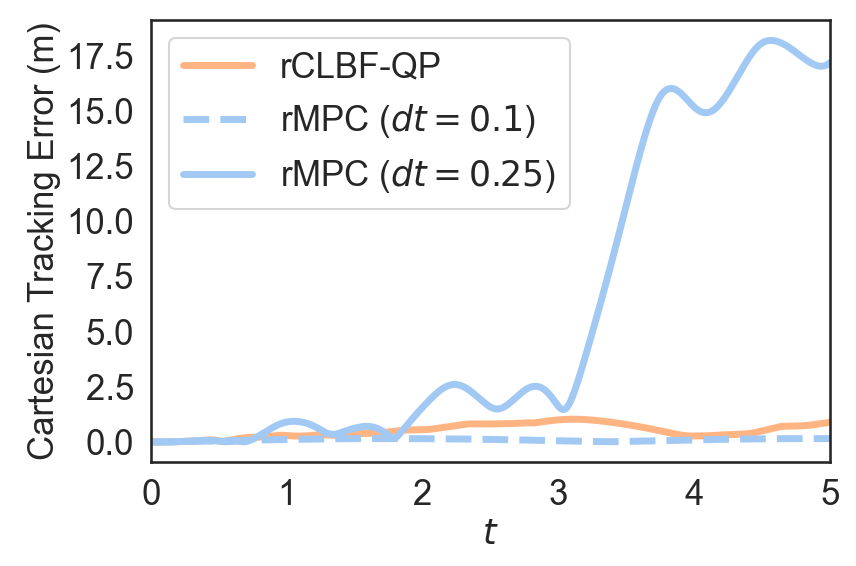}
         \includegraphics[height=2.5cm]{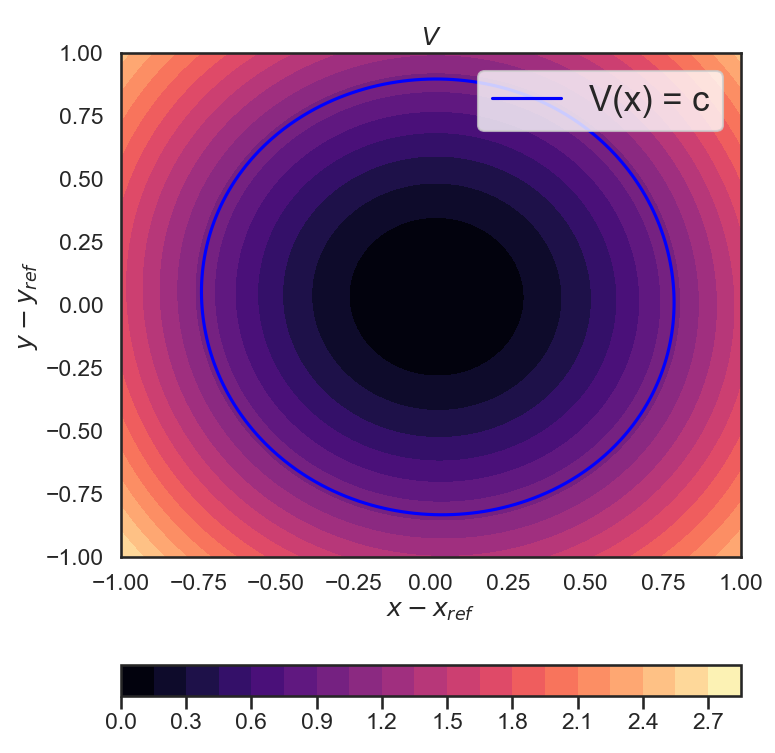}
     \end{subfigure}
     \end{minipage}
         
     
    \caption{Trajectory tracking on kinematic (left) and sideslip (right) vehicle models, with contour plots of $V$. Blue shows the $c$-level set.}
    \label{fig:car_tracking}
\end{figure}
\vspace{-1.5em}
\subsection{UAV stabilization}

The next two examples involve stabilizing a quadrotor near the ground while maintaining a minimum altitude. Relative to the previous examples, these benchmarks increase the complexity of the state constraints, and we consider two models with increasingly challenging dynamics. The first model (referred to as the ``3D quadrotor'') has 9 state dimensions for position, velocity, and orientation, with control inputs for the net thrust and angular velocities \cite{Sun2020}. The second model (the ``neural lander'') has lower state dimension, including only translation and velocity, with linear acceleration as an input, but its dynamics include a neural network trained to approximate the aerodynamic ground effect, which is particularly relevant to this safe hovering task \cite{liu2020robust}. The mass of both models is uncertain, but assumed to lie on $[1.0, 1.5]$ for the 3D quadrotor and $[1.47, 2.0]$ for the neural lander.

Fig.~\ref{fig:uav_stabilization} shows simulation results on these two models. The trend from the previous benchmarks continues: our controller maintains safety while reducing evaluation time by a factor of 10 relative to MPC. Moreover, while the robust MPC method can achieve low error relative to the goal for the the 3D quadrotor model, the nonlinear ground effect term prevents MPC from driving the neural lander to the goal. In contrast, the rCLBF-QP method can consider the full nonlinear dynamics of the system, including the learned ground effect, and achieves a much lower error relative to the goal.

\begin{figure}[bh!]
     \centering
     \hspace{-2cm}
     \begin{subfigure}[t]{1.2\linewidth}
         \includegraphics[height=2.9cm]{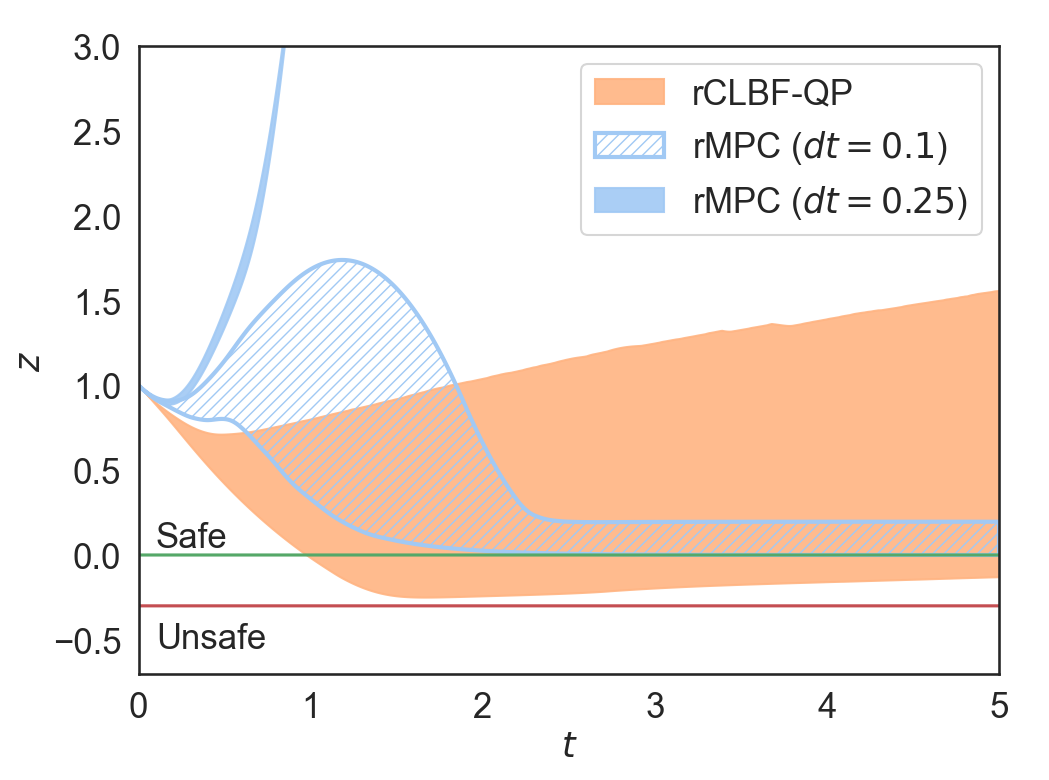}
         \includegraphics[height=2.9cm]{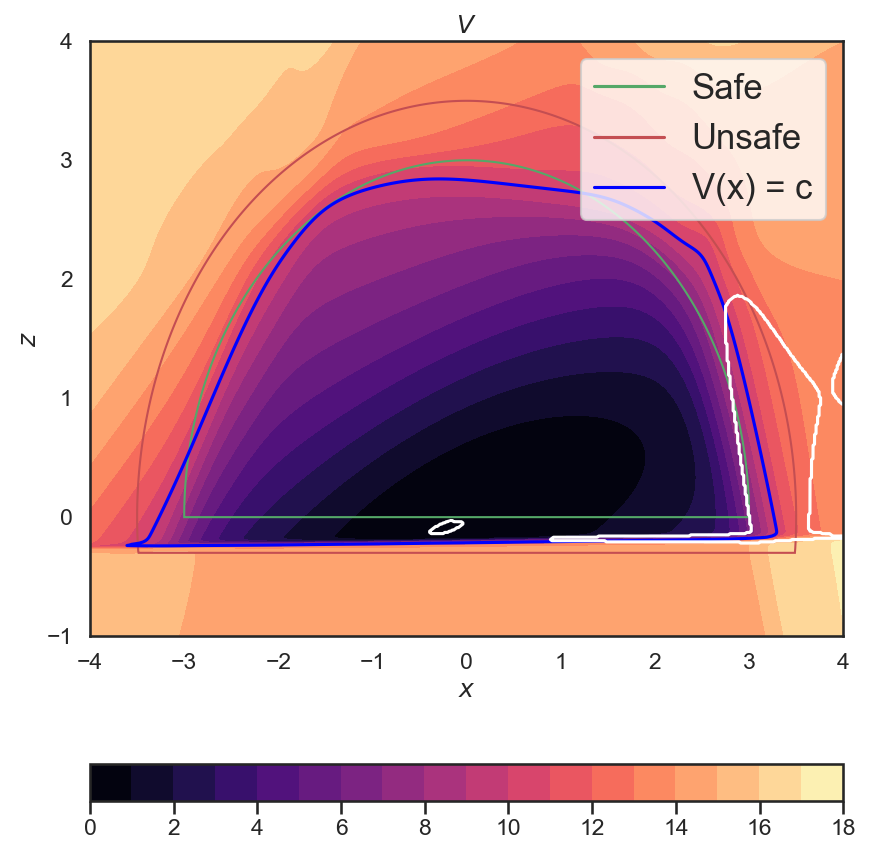}
         \includegraphics[height=2.9cm]{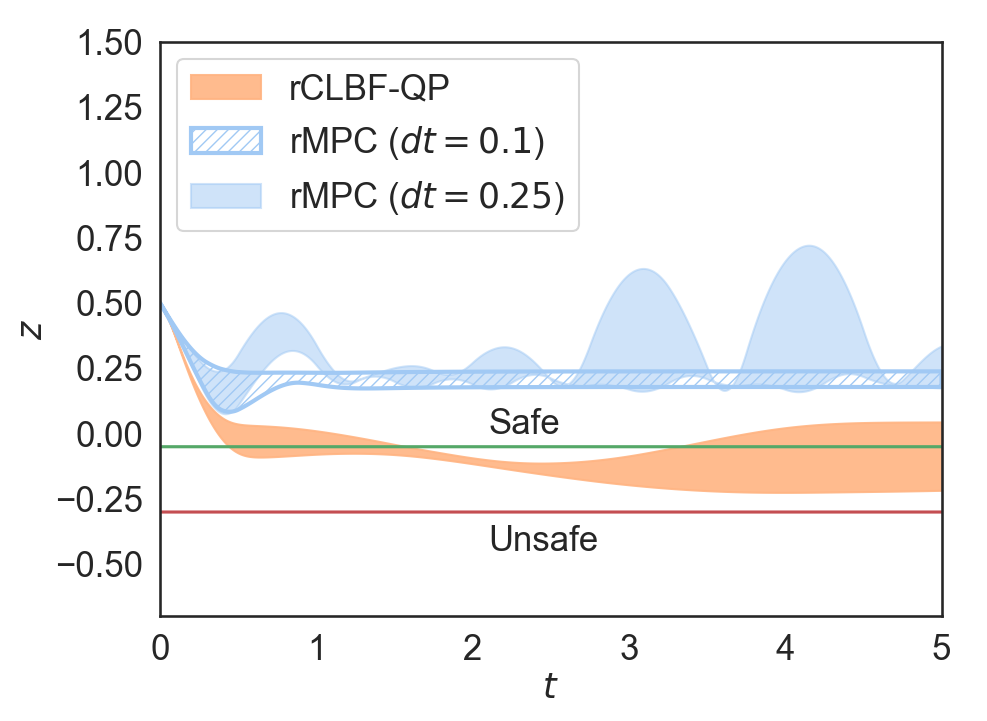}
         \includegraphics[height=2.9cm]{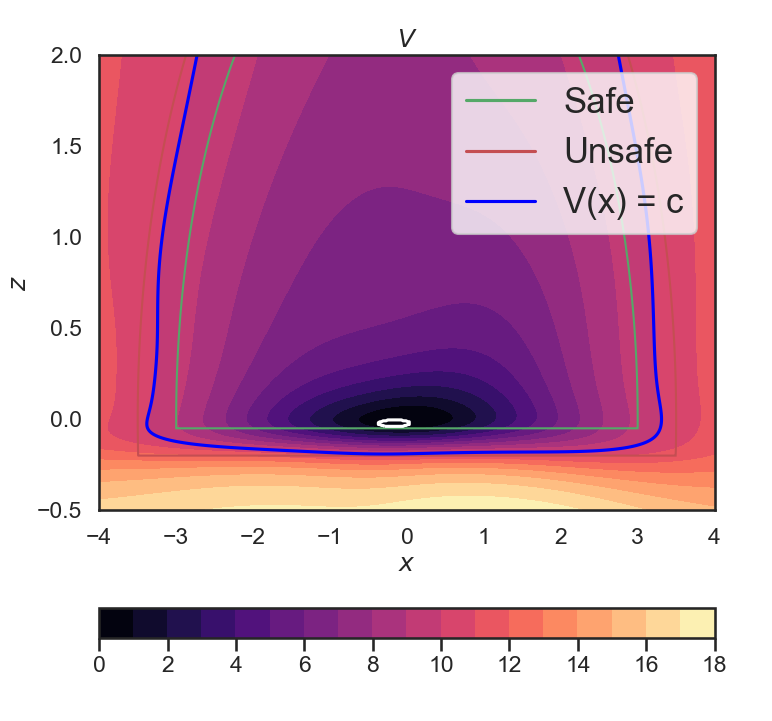}
     \end{subfigure}
    \caption{Controller performance for the 3D quadrotor (left) and neural lander (right), with contour plots of $V$. Blue shows the $c$-level set, white shows regions where condition~\eqref{eq:rclbf_decrease_cond} is violated.}
    \label{fig:uav_stabilization}
\end{figure}
\vspace{-1.5em}
\subsection{Navigation with non-convex safety constraints}

The preceding benchmarks all include convex safety constraints that can be easily encoded in a linear robust MPC scheme. Our next set of examples demonstrate the ability of our approach to generalize to complex environments. These problems are commonly solved by combining planning and robust tracking control, so in our comparisons we use robust MPC to track a safe reference path through each environment. In contrast, our rCLBF-QP controller is not provided with a reference path and instead synthesizes a safe controller using only the model dynamics and (non-convex) safety constraints, which is a more challenging problem than the tracking problem as in Section~\ref{sec:tracking}.  The three navigation problems we consider are: (a) controlling a Segway to duck under an obstacle to reach a goal \cite{aastrom2021feedback}, (b) navigating a 2D quadrotor model around obstacles \cite{Sun2020}, and (c) completing a satellite rendezvous that requires approaching the target satellite from a specific direction \cite{jewison2016spacecraft}. 
For (a) and (c), we conducted additional comparisons with a Hamilton-Jacobi-based controller (HJ,~\cite{toolboxls}) and policy trained via constrained policy optimization reinforcement learning (CPO,~\cite{cpo}).
Simulated trajectories are shown in Fig.~\ref{fig:navigation}. Note that in the Segway and satellite examples, robust MPC fails to track the reference path, while the rCLBF controller successfully navigates the environment.
HJ preserves safety in the satellite example but fails to reach the goal (which is positioned near the border of the unsafe region), while HJ controller synthesis failed in the Segway example (the backwards reachable set did not reach the start location with a \SI{5}{s} horizon). Note that the HJ satellite controller requires different initial conditions, since it will fail if started outside of the safe region. The policy trained using CPO navigates to the goal in the satellite example, but it is not safe. In the Segway example, CPO does not learn a stable controller (details are given in the appendix).
\begin{figure}[bh!]
     \centering
     \hspace{-1cm}
     \begin{minipage}[c][6.5cm][t]{.65\textwidth}
      \vspace*{\fill}
      \begin{subfigure}[b]{\linewidth}
         \centering
         \includegraphics[height=2.5cm]{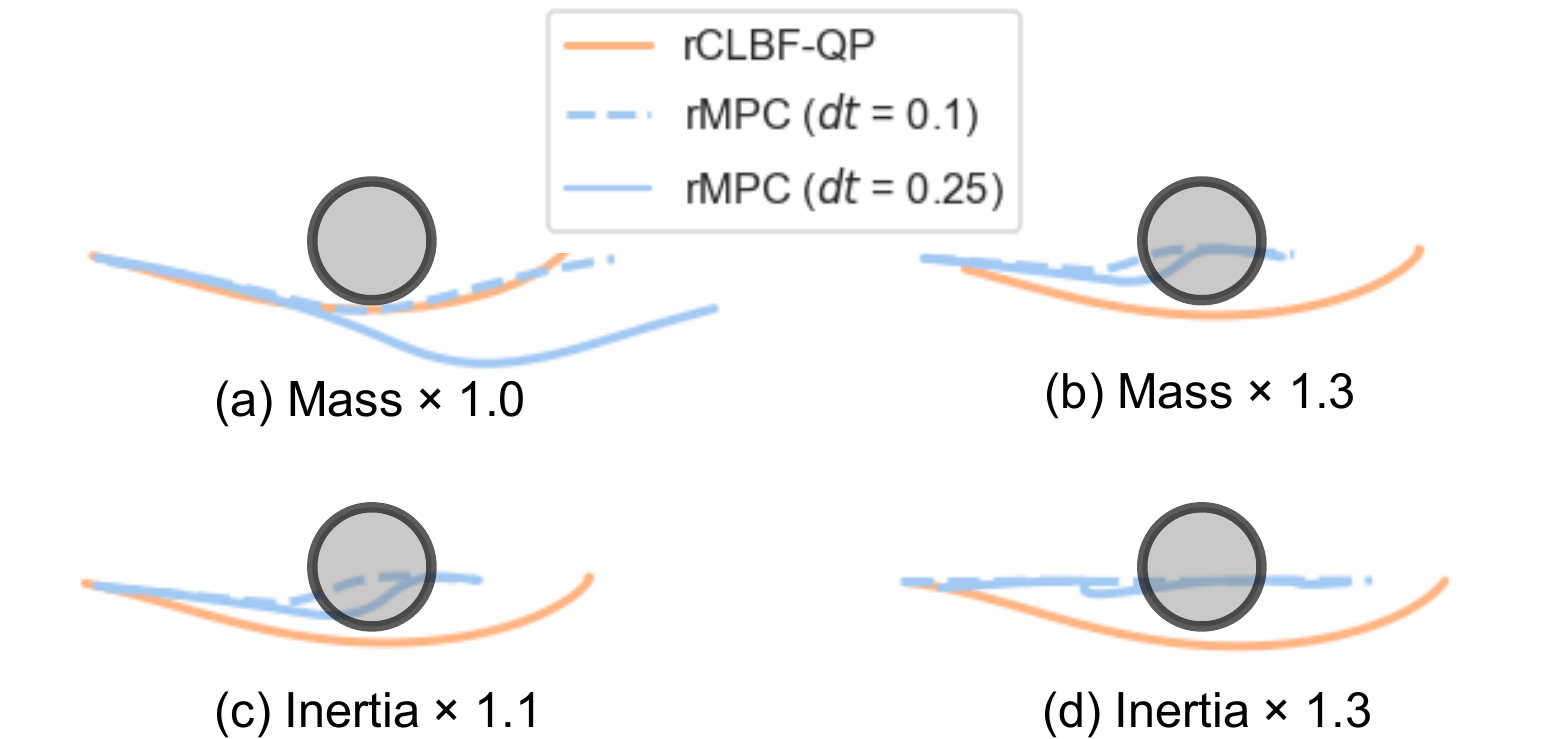}
         \label{fig:segway}
     \end{subfigure}\par\vfill
     \begin{subfigure}[b]{\linewidth}
         \centering
         \includegraphics[height=3.5cm]{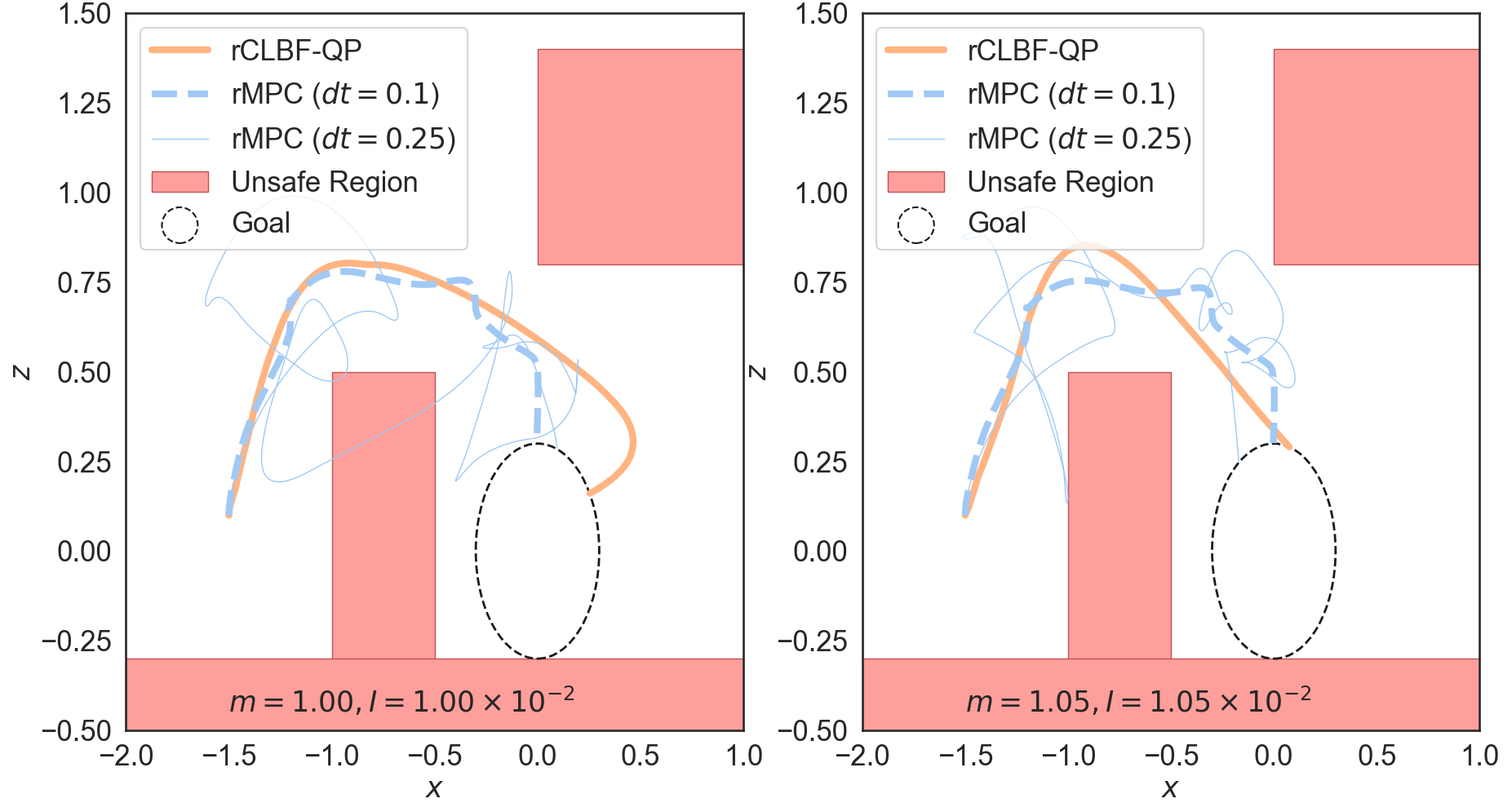}
         \label{fig:quad2d}
     \end{subfigure}
    \end{minipage}%
    \begin{minipage}[c][6.5cm][t]{.4\textwidth}
      \vspace*{\fill}
      \begin{subfigure}[b]{\linewidth}
         \centering
         \includegraphics[height=6cm]{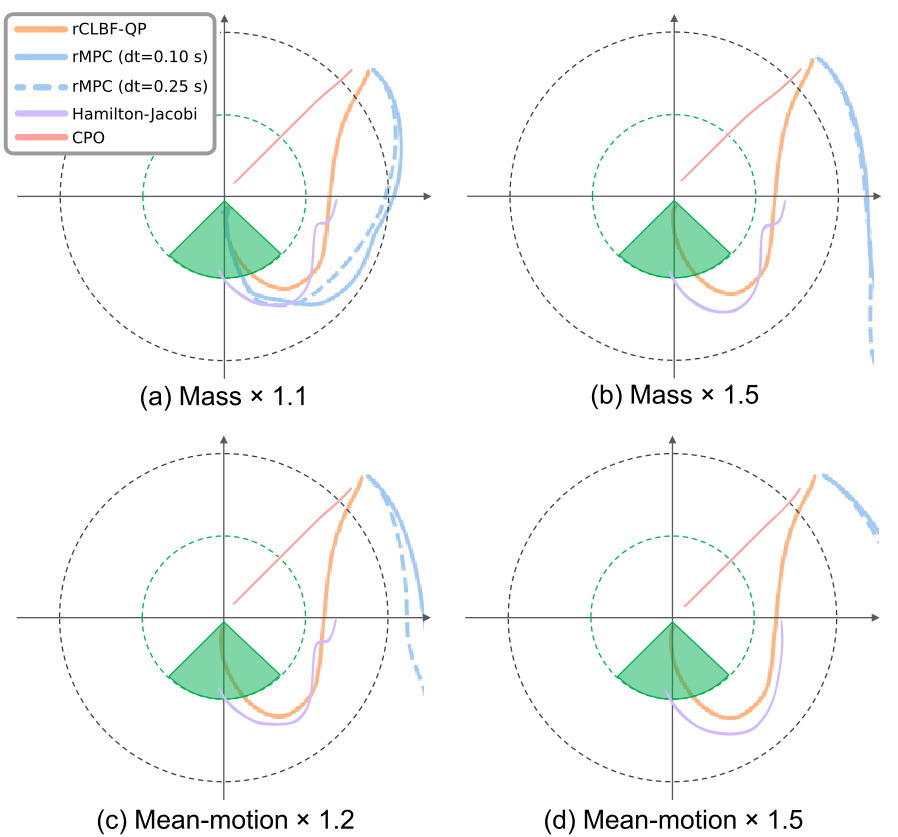}
         \label{fig:satellite}
      \end{subfigure}
    \end{minipage}
    \caption{Navigation problems solved using our rCLBF-QP controller, compared with robust MPC. Clockwise from right: satellite rendezvous, planar quadrotor, and Segway.}
    \label{fig:navigation}
\end{figure}
\vspace{-2em}
\section{Discussion \& Conclusion}
These results demonstrate two clear trends. First, the performance of our controller (in terms of both safety rate and error relative to the goal) is comparable to that of MPC when the MPC controller is stable. In some cases, our method achieves lower steady-state error due to its ability to consider highly nonlinear dynamics, as in the neural lander example. In other cases, the dynamics are well-approximated by the linearization and robust MPC achieves better steady-state error, but our approach still achieves a comparable safety rate. Second, we observe that the performance of the robust MPC algorithm is highly sensitive to the control frequency, and these controllers are only stable at control frequencies that cannot run in real-time on a laptop computer. This highlights one benefit of our method over traditional MPC, which trades increased offline computation for an order of magnitude reduction in evaluation time. In all cases, we find that our proposed algorithm finds a controller that satisfies the safety constraints despite variation in model parameters, validating our claim of presenting a framework for robust safe controller synthesis. 

In summary, we present a novel, learning-based approach to synthesizing robust nonlinear feedback controllers. Our approach is guided by a robust extension to the theory of control Lyapunov barrier functions that explicitly accounts for uncertainty in model parameters. Through experiments in simulation, we successfully demonstrate the performance of our approach on a range of challenging safe control problems. A number of interesting open questions remain, including scalable verification strategies for $V$, the sample complexity of this learning method, and the relative convergence rates of $V$, $\pi_{\mathrm{NN}}$, and the QP controller derived from $V$, which we hope to revisit in future work. We also plan on exploring application to hardware systems, including considerations of delay and state estimation uncertainty.

\section*{Acknowledgments}

The NASA University Leadership Initiative (grant \#80NSSC20M0163) and Defense Science and Technology Agency in Singapore provided funds to assist the authors with their research, but this article solely reflects the opinions and conclusions of its authors and not any NASA entity, DSTA Singapore, or the Singapore Government.
C. Dawson is supported by the NSF Graduate Research Fellowship under Grant No. 1745302.

\bibliography{references}

\newpage
\appendix

\section*{Supplementary Materials}

In addition to the sections below, we include a video demonstrating our controller's performance on the kinematic car trajectory tracking and 2D quadrotor obstacle avoidance benchmarks. In addition, we include documented code for running several of our examples.

\section*{Proof of Theorem~\ref{clbf_both}}

The proof of Theorem~\ref{clbf_both} follows from the following lemmas, which prove stability and safety of CLBF-based control, respectively.

\begin{lemma}\label{clbf_stability}
    If $V(x)$ is a CLBF, then any control policy $\pi(x) \in K(x) \ \forall\ x \in \cX$ will exponentially stabilize the system $\dot{x} = f(x) + g(x)\pi(x)$ to $\xg$.
\end{lemma}
\begin{proof}
    Since $\pi(x) \in K(x)$, it follows that $\der{V}{t} \leq -\lambda V(x)$ for the closed loop system. Thus, $V$ is a Lyapunov function and proves exponential stability about $\xg$.
\end{proof}
\begin{lemma}\label{clbf_safety}
    If $V(x)$ is a CLBF, then for any control policy $\pi(x) \in K(x) \ \forall\ x \in \cX$ and any initial condition $x(0) \in \xs$, $x(t) \notin \xu\ \forall\ t > 0$ (i.e. any trajectory starting in the safe set will never enter the unsafe region).
\end{lemma}
\begin{proof}
    For convenience, define $\mathcal{V} = V \circ x(t)$. Since $x(0) \in \xs$, condition~\eqref{eq:clbf_safe_cond} implies that $\mathcal{V}(0) \leq c$. Conditions~\eqref{eq:clbf_pos_cond} and~\eqref{eq:clbf_decrease_cond} ensure that $\mathcal{V}$ is strictly decreasing in time (except when $x(t) = \xg$, at which point $\mathcal{V}$ is constant at zero). As a result, $\mathcal{V}(t) < \mathcal{V}(0) \leq c\ \forall\ t > 0$. If $x(t)$ were to enter the unsafe region, there would exist $t_u > 0$ such that $\mathcal{V}(t_u) > c$. This is a contradiction, so we conclude that $x(t)$ will never enter the unsafe region for $t > 0$.
\end{proof}

\section*{Proof of Theorem~\ref{rclbf}}

\begin{proof}
    By assumption, $f_\theta$ and $g_\theta$ are affine in $\theta$. Additionally, the Lie derivatives $L_f V$ and $L_g V$ are affine in $f$ and $g$, and the rCLBF constraint \eqref{eq:rclbf_qp_cons} is affine in $L_f V$ and $L_g V$. As a result, the overall mapping from $\Theta$ to the left-hand side of \eqref{eq:rclbf_qp_cons} is affine and thus maps the convex hull of $\theta_0,\ldots,\theta_{n_s}$ to the convex hull of $L_{f_{\theta_0}} V + L_{g_{\theta_0}} V u + \lambda V,\ldots,L_{f_{\theta_{n_s}}} V + L_{g_{\theta_{n_s}}} V u + \lambda V$. It follows that if \eqref{eq:rclbf_qp_cons} is satisfied for each scenario $\theta_i$ then it will be satisfied for any possible $\theta \in \Theta$. We can conclude that the rCLBF satisfies the conditions of a standard CLBF for any particular realization of the system with parameters $\theta \in \Theta$, so the safety and stability results of Theorem~\ref{clbf_both} apply.
\end{proof}

\section*{Implementation of Learning Approach}

In this section, we describe several details of our implementation of the system used to train $V$ and $\pi_{\mathrm{NN}}$. At a high level, our system is implemented in PyTorch \cite{pytorch} using PyTorch Lightning \cite{falcon2019pytorch}. All neural networks were implemented with $\tanh$ activation functions, and we used batched stochastic gradient descent with weight decay for optimization (with learning rate $10^{-3}$ and decay rate $10^{-6}$). The next paragraphs describe our training strategies.

\textbf{Sampling of training data:} we found that training performance could be improved by specifying a fixed percentage of training points that must be sampled from the goal, safe, and unsafe regions. For example, instead of sampling $N_{train}$ points uniformly from the state space, we might sample $0.1 N_{train}$ uniformly from the goal region, $0.1 N_{train}$ uniformly from the unsafe region, $0.1 N_{train}$ uniformly from the safe region, and the remaining $0.7 N_{train}$ uniformly from the entire state space.

\textbf{Network initialization: } although it was not necessary for all experiments, we found that some experiments (particularly the car trajectory tracking benchmarks) performed better if the CLBF network was initialized to match the quadratic Lyapunov function found by linearizing the system about the goal point. After training $V$ for several epochs to match this quadratic initial guess, we then alternated between training $V$ and $u_\mathrm{NN}$, optimizing one for several epochs before optimizing the other. We found that on some examples this stabilized the learning process. We did not notice an improvement from episodic learning, although this may be more useful when training on higher-dimensional systems.

\textbf{Hyperparameter tuning: } during the development process, we optimized hyperparameters ($c,\lambda,\epsilon$, the size of the $V$ and $u_\mathrm{NN}$ networks, and the penalty applied to relaxations of the QP constraints) based on a combination of the empirical loss on a test data set and through controller performance in simulation. In most experiments, we found that $c=1$ and $\lambda \in [0.1, 10]$ were sensible defaults, along with neural networks with 2 hidden layers of 64 units each. We found that tuning parameters $a_1=a_2=100$ and $a_3=1$ yield controllers that perform well in simulation.

\textbf{Reach-avoid problem specification: } when defining reach-avoid problems for this approach, care should be taken when specifying $\xs$ and $\xu$. We found that it is necessary to have some region in between the safe and unsafe sets where the neural rCLBF has the freedom to adjust the boundary at $V(x) = c$ as needed to find a valid rCLBF. In addition, we found that including a safety constraint that prevents the system from leaving the region where training data was gathered improves the controller's performance.

\textbf{rCLBF-QP Relaxation: } to ensure that the controller is always feasible, we permit the QP to relax the constraints on the CLBF derivative, and the extent of this relaxation is penalized with a large coefficient in the QP objective. The penalty coefficients used in different experiments are included below. This relaxation also provides a useful training signal for the $V$ network. To make use of this signal, we solve \eqref{eq:rclbf_qp} for each point at training-time and scale the last term of the loss function point-wise by the relaxation, effectively increasing the penalty for regions where the feasible set of \eqref{eq:rclbf_qp} is empty and decreasing the penalty in regions where there exists a feasible solution (even if $\pi_{\rm{NN}}$ has not yet converged to find that feasible solution).

\section*{Verification of Learned CLBFs}

Our focus in this paper is primarily on the use of robust CLBFs to automatically synthesize feedback controllers for nonlinear safe control tasks. We find that our learning method yields functions that satisfy the rCLBF conditions in the vast majority of the state space, and yields feedback controllers that are successful in simulation, but we do not claim to have exhaustively verified our learned rCLBFs. Indeed, scalable verification for learned certificate functions remains an open problem. Relevant verification techniques include neural network reachability analysis (see \cite{Liu2021} for a recent survey), SMT solvers \cite{Chang2019}, Lipschitz-informed sampling methods \cite{Bobiti2018}, and probabilistic claims from learning theory \cite{Qin2021}.

Additionally, these verification techniques might be used in future work to inform the training of an rCLBF neural network. For instance, spectral normalization \cite{miyato2018spectral} of the rCLBF network would allow us to tune the Lipschitz constant of $V(x)$, enabling more effective use of Lipschitz-informed sampling verification tools. Similarly, reachability tools and SMT solvers can provide counter-examples to augment the training data and make further failures less likely \cite{Chang2019}. Further, almost Lyapunov functions~\cite{liu2020almost,Boffi2020} show that even if the Lyapunov conditions do not hold everywhere the system is still provably stable; this result may generalize to CLBFs as well. These are all exciting directions that we hope to explore in our future work on this topic.

\section*{Implementation of Robust MPC}

We implemented our robust MPC scheme in Matlab following the example in the YALMIP documentation \cite{lofberg2012}, which is in turn based on the algorithm published in \cite{lofberg2012}. This MPC algorithm relies on a linearization of the system dynamics, and we used a constant linearization about the goal state. For trajectory tracking examples, we linearize the system about the reference trajectory.

The robust MPC problem was formulated in YALMIP and Gurobi \cite{gurobi} was used as the underlying QP solver. When measuring evaluation times for robust MPC, we first use YALMIP to pre-compile the robust QP then measure the time needed to solve the compiled QP using Matlab's built-in \texttt{timeit} function. We understand that additional optimizations (e.g. explicit MPC) might reduce the evaluation time of robust MPC further, but those optimizations can be applied equally well to speeding up the QP solution in our proposed controller. Effectively, for the purposes of measuring performance, we optimize both approaches to the point where a single quadratic program is being sent to the Gurobi QP solver, and so we believe we have provided a fair comparison in our results.

\section*{Implementation of Hamilton Jacobi Control Synthesis}

To compute the Hamilton-Jacobi value function, we used the helperOC package at \url{https://github.com/HJReachability/helperOC}, which wraps the toolboxLS software \cite{toolboxls}. We over-approximate the parametric uncertainty with an additive uncertainty. In the Segway example, where the unsafe set is defined in terms of $(x, y)$, we over-approximate this unsafe set using a polytope defined on $(p, \theta)$. We computed the HJ value function, then applied the optimal HJ controller forwards described in \cite{hj_overview}. We used a time step of 0.05 seconds and a maximum horizon of 5 seconds while computing the backwards reachable set. The HJ value function was approximated on a grid, and the grid resolution was set to balance accuracy and running time.

\section*{Details on Simulation Experiments}

This section reports the dynamics and hyperparameters used in our experiments. Note that in some of our examples, mass is an uncertain parameter but enters into the dynamics as $1/m$ (similarly for rotational inertia). In these cases we treat $1/m$ as the uncertain parameter and proceed with our method as described in Section~\ref{learning_approach}. For clarity, we give the uncertainty ranges in terms of $m$ rather than in terms of the reciprocal.

Training was conducted on a workstation with a 32-core AMD 3970X CPU and four Nvidia GeForce 2080 Ti GPUs (one GPU was used for each training job, allowing us to parallelize our experiments). Runtime evaluation was conducted on a consumer laptop with an Intel i7-8565U CPU running at 1.8 GHz, and no GPU.

\subsection*{Kinematic Car}

We use the kinematic single track model of a car given in the CommonRoad benchmarks \cite{Althoff2017a}. We modify this model to express position and orientation relative to a reference path parameterized by $v_{ref}$, $a_{ref}$, $\psi_{ref}$, and $\omega_{ref}$ (the linear velocity and acceleration, angle, and angular velocity of the reference path). To model a reference path with uncertain curvature, we treat $\omega_{ref}$ as the uncertain parameter and assume that it vares on $[-1.5, 1.5]$.

The state of the path-centric kinematic car model is $[x_e, y_e, \delta, v_e, \psi_e]$, representing Cartesian error, steering angle, velocity error, and heading error, and the control inputs are $v_\delta$ and $a_{long}$ (the steering angle velocity and longitudinal acceleration). The dynamics are given by $\dot{x} = f(x) + g(x) u$, with
\begin{align}
    f(x) &= \mat{
        v \cos(\psi_e) - v_{ref} + \omega_{ref} * y_e \\
        v \sin(\psi_e) - \omega_{ref} x_e \\
        0 \\
        -a_{ref} \\
        \frac{v}{l_r + l_f} \tan(\delta) - \omega_{ref}
    } \\
    g(x) &= \mat{
        0 & 0 \\
        0 & 0 \\
        1 & 0 \\
        0 & 1 \\
        0 & 0
    }
\end{align}
where we define $v = v_e + v_{ref}$ and $l_f$ and $l_r$ are vehicle parameters measuring the distance from the center of mass to the front and rear axles (these parameters are taken from the CommonRoad \texttt{vehicle-2} benchmark).

We define a goal point $\xg$ as the origin with nominal parameters $v_{ref} = 10.0$, $a_{ref} = 0.0$, $\omega_{ref} = 0.0$ (note that the reference heading and reference position do not enter directly into the dynamics). These tracking tasks are not reach-avoid tasks, as there is no hard constraint other than maintaining bounded tracking error. We used the LQR solution with nominal parameters for $\pi_{nominal}$. Training data were sampled from $x_e, y_e, v_e \in [-3, 3]$, $\psi_e \in [-\pi / 2, \pi/2]$, and $\delta \in [-1.066, 1.066]$, but we selectively re-sampled until at least $40\%$ of the data were within $\norm{x} \leq 1$, at least $20\%$ were within $\norm{x} \leq 0.25$, and at least $20\%$ were $\norm{x} \geq 1.5$, which ensured that adequate training data were sampled from near the goal point. 125,000 samples were used for training, with 10\% reserved for validation. $V$ and $\pi_\mathrm{NN}$ are parameterized as two-layer fully-connected neural networks with hidden layer size of 64 and $\mathrm{tanh}$ activation. We set $c=1$, $\lambda=1$, and allowed relaxations of the constraints in \eqref{eq:rclbf_qp} with penalty coefficient $10$.

A contour plot of the learned $V$ is shown in Fig.~\ref{fig:kscar_V} as a function of $x_e$ and $y_e$, with all other state variables zero. From this plot, we see that some violation of \eqref{eq:rclbf_decrease_cond} occurs near the origin (the violation was computed on a grid with maximum spacing of $0.008$ between points). This makes sense, as this system is likely impossible to robustly stabilize around the origin (i.e. we suspect that there is no fixed $u$ that renders the origin a fixed point for any $\omega_{ref}$). Outside the origin, we see that the CLBF conditions are satisfied, which agrees with what we observe in simulation: our controller leaves the origin but then stabilizes with a relatively constant tracking error.

\begin{figure}[h]
    \centering
    \includegraphics[width=0.8\linewidth]{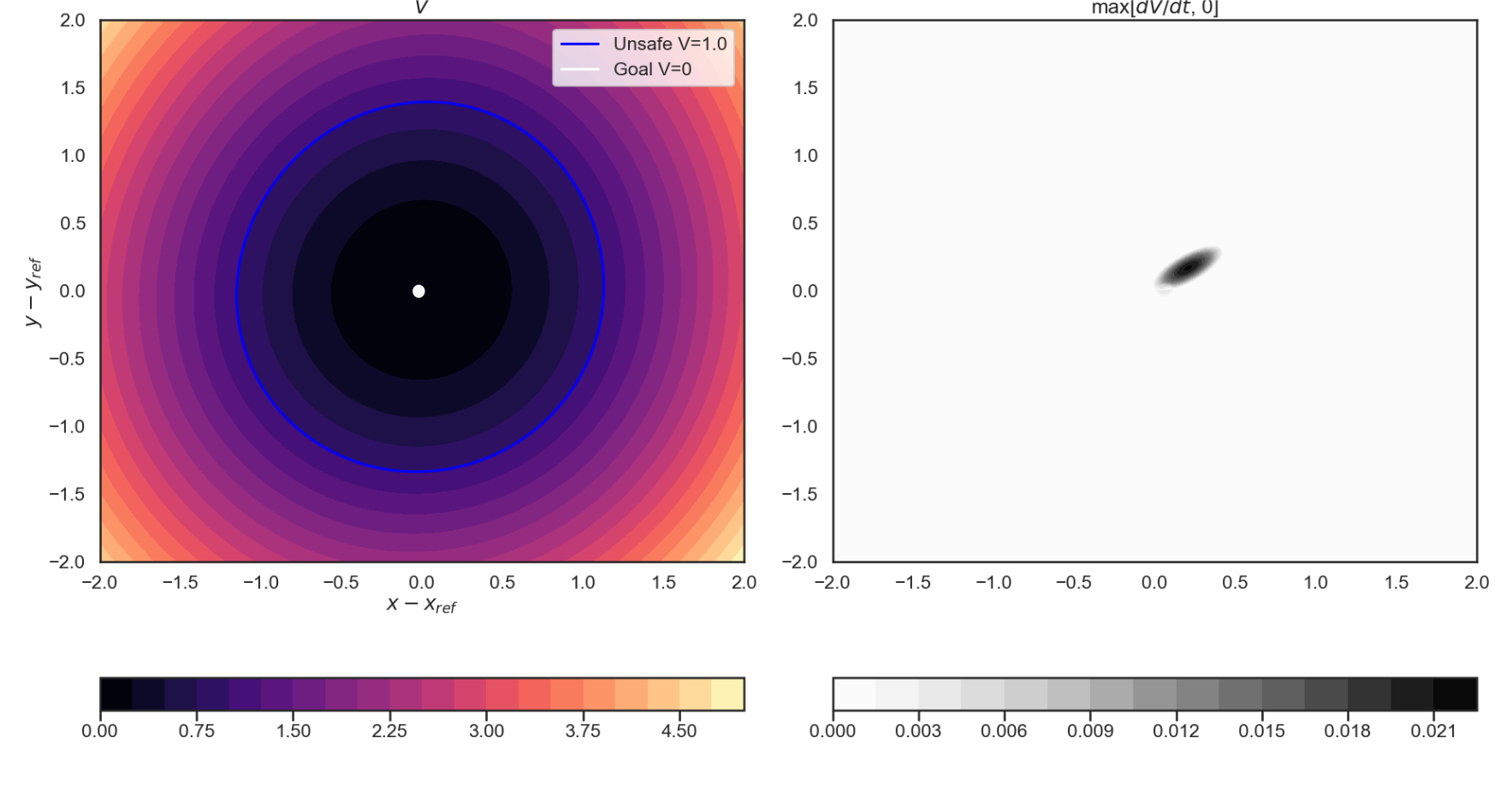}
    \caption{A contour plot of the learned rCLBF $V$ (right) and violation of condition \eqref{eq:rclbf_decrease_cond} (left) for the kinematic car tracking task. The violation of the rCLBF decrease condition \eqref{eq:rclbf_decrease_cond}, which was found to be at most $0.0225$ over this range, was computed as $\max(dV/dt, 0)$, summed over both parameter scenarios.}
    \label{fig:kscar_V}
\end{figure}

\subsection*{Car with Sideslip}

We use the single-track model given in the CommonRoad benchmarks \cite{Althoff2017a}. Similarly to the kinematic model, we express the dynamics in path-centric coordinates and treat $\omega_{ref}$ (which determines the curvature of the reference path) as the uncertain parameter and assume that it vares on $[-1.5, 1.5]$.

Compared to the kinematic model, the single-track (sideslip) model has two additional state variables: $\beta$ (the sideslip angle) and $\dot{\psi}_e$ (the rate of change of the heading error). The control inputs are the same as for the kinematic model. The dynamics are given by $\dot{x} = f(x) + g(x) u$, with
\begin{align}
    f(x) &= \mat{
        v \cos(\psi_e) - v_{ref} + \omega_{ref} * y_e \\ 
        v \sin(\psi_e) - \omega_{ref} x_e \\ 
        0 \\ 
        0 \\ 
        \dot{\psi}_e \\ 
        -\frac{\mu m}{v I_z (l_r + l_f)}(l_f^2 C_{Sf} g l_r + lr^2 C_{Sr} g l_f) (\dot{psi} + \omega_{ref}) + \frac{\mu m}{I_z (l_r + l_f)} (l_r C_{Sr} g l_f - l_f C_{Sf} g l_r) \beta + \frac{\mu m}{I_z (l_r + l_f)} (l_f C_{Sf} g l_r) \delta \\ 
        (\frac{mu}{v^2 (l_r + l_f)} (C_{Sr} g l_f l_r - C_{Sf} g l_r l_f) - 1) (\dot{psi}_e - \omega_{ref}) - \frac{\mu}{v (l_r + l_f)}(C_{Sr} g l_f C_{Sf} g l_r) \beta
            + \frac{\mu}{v (l_r + l_f)} (C_{Sf} g l_r) * \delta \\ 
    } \\
    g(x) &= \mat{
        0 & 0 \\
        0 & 0 \\
        1 & 0 \\
        0 & 1 \\
        0 & 0
    }
\end{align}
where we define $v = v_e + v_{ref}$ and $l_f, l_r, C_{Sf}, C_{Sr}, \mu$ are parameters whose values taken from the CommonRoad \texttt{vehicle-2} benchmark ($g$ is gravitational acceleration). Since these dynamics become singular at low speeds, for $|v| < 0.1$ we revert to the kinematic model (as described in \cite{Althoff2017a}).

We define a goal point $\xg$ as the origin with nominal parameters $v_{ref} = 10.0$, $a_{ref} = 0.0$, $\omega_{ref} = 0.0$ (note that the reference heading and reference position do not enter directly into the dynamics). These tracking tasks are not reach-avoid tasks, as there is no hard constraint other than maintaining bounded tracking error. We used the LQR solution with nominal parameters for $\pi_{nominal}$. Training data were sampled from $x_e, y_e, v_e \in [-3, 3]$, $\psi_e \in [-\pi / 2, \pi/2]$, $\dot{\psi}_e \in [-\pi/2, \pi/2]$, $\delta \in [-1.066, 1.066]$, and $\beta \in [-\pi/3, \pi/3]$, but we selectively re-sampled until at least $40\%$ of the data were within $\norm{x} \leq 0.35$, at least $20\%$ were within $\norm{x} \leq 0.25$, and at least $20\%$ were $\norm{x} \geq 0.85$, which ensured that adequate training data were sampled from near the goal point. 125,000 samples were used for training, with 10\% reserved for validation. $V$ and $\pi_\mathrm{NN}$ are parameterized as two-layer fully-connected neural networks with hidden layer size of 64 and $\mathrm{tanh}$ activation. We set $c=1$, $\lambda=0.1$, and allowed relaxations of the constraints in \eqref{eq:rclbf_qp} with penalty coefficient $10^{8}$.

When we examine the contour plot of the learned $V$ (shown in Fig.~\ref{fig:stcar_V} as a function of $x_e$ and $y_e$, with all other state variables zero), we see that it has a similar shape as that learned for the kinematic model, but we did not detect any violation of \eqref{eq:rclbf_decrease_cond} on a grid with maximum spacing of $0.004$ between points. However, given our controller's simulation performance on this task, which included some error relative to the goal, it is likely that some violation of the CLBF conditions would be seen on other cross sections (i.e. where state variables other than $x_e$ and $y_e$ are varied).

\begin{figure}[h]
    \centering
    \includegraphics[width=0.8\linewidth]{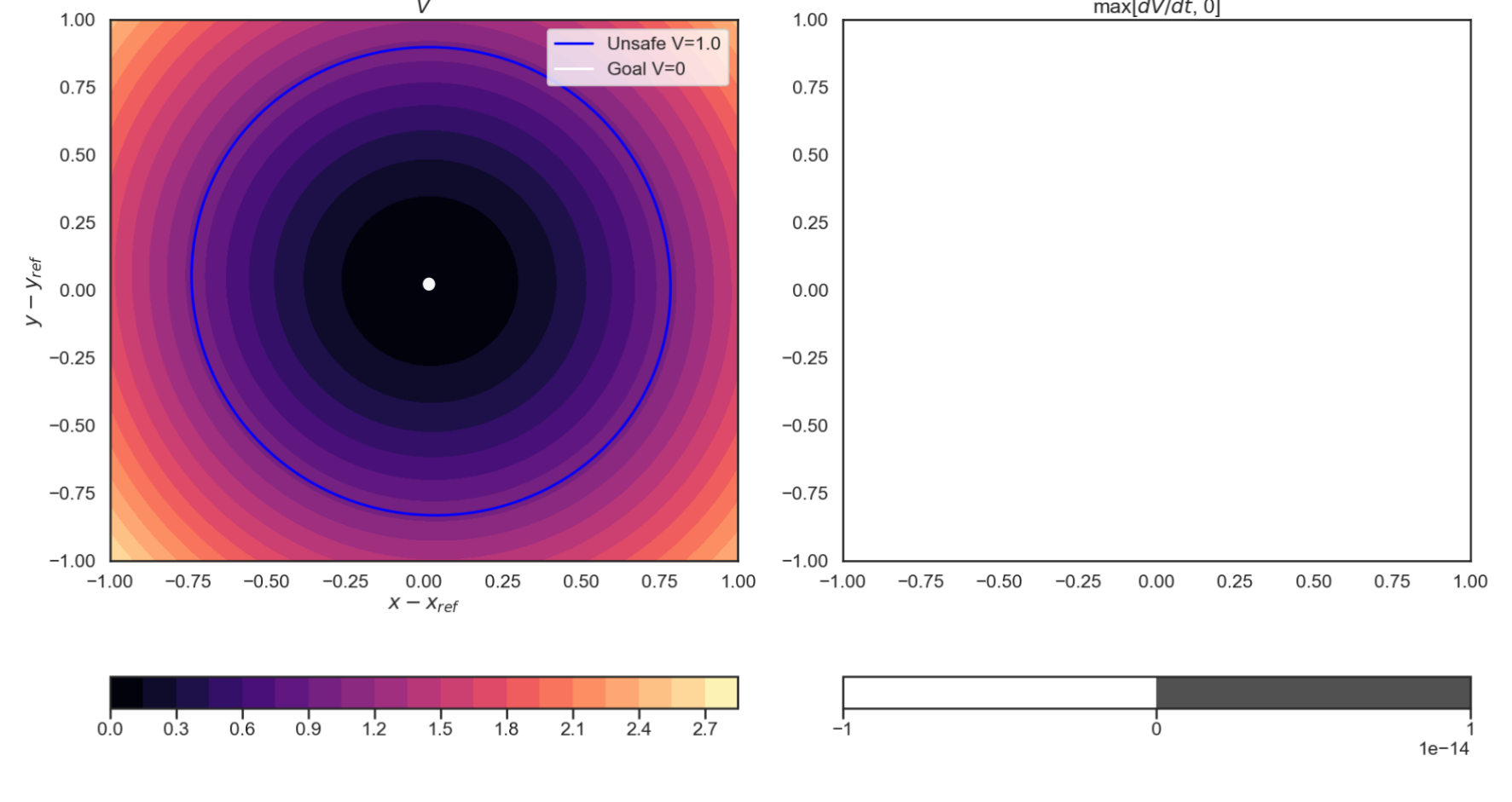}
    \caption{A contour plot of the learned rCLBF $V$ (right) and violation of condition \eqref{eq:rclbf_decrease_cond} (left) for the sideslip car tracking task. The rCLBF decrease condition \eqref{eq:rclbf_decrease_cond}, computed as $\max(dV/dt, 0)$, summed over both parameter scenarios, was not found to be violated on this range.}
    \label{fig:stcar_V}
\end{figure}

\subsection*{3D Quadrotor}

The state of the 9-dimensional quadrotor is given by $x = [p_x, p_y, p_z, v_x, v_y, v_z, \phi, \theta, \psi]$, with control vector $u = [f, \dot{\phi}, \dot{\theta}, \dot{\psi}]$. This model is adapted from \cite{Sun2020}. The system is parameterized by mass $m$. The dynamics are given by $\dot{x} = f(x) + g(x) u$, with

\begin{align}
    f(x) &= \mat{v_x, v_y, v_z, 0, 0, -g, 0, 0, 0}^T \\
    g(x) &= \mat{
        0 & 0 & 0 & 0 \\
        0 & 0 & 0 & 0 \\
        0 & 0 & 0 & 0 \\
        -(1/m)\sin\theta & 0 & 0 & 0 \\
        (1/m)\cos\theta\sin\phi & 0 & 0 & 0 \\
        (1/m)\cos\theta\cos\phi & 0 & 0 & 0 \\
        0 & 1 & 0 & 0 \\
        0 & 0 & 1 & 0 \\
        0 & 0 & 0 & 1
    }
\end{align}
where $g$ is the gravitational acceleration. Note that although these dynamics are not affine in $m$, they are affine in $1/m$, which can be treated as the uncertain parameter without loss in generality.

In this task, $\xg$ is the origin, $\xs = \set{x\ :\ p_z \geq 0 \land \norm{x} \leq 3}$, and $\xu = \set{x\ :\ p_z \leq -0.3 \lor \norm{x} \geq 3.5}$. To model uncertainty in the mass of the quadrotor's payload, we simulate both the rCLBF-QP and MPC controllers with masses sampled uniformly from $m \in [1.0, 1.5]$. To isolate the impact of parameter variation on controller performance, we use a constant initial condition $x(0) = [1, 1, 1, 1, 1, -1, 1, 1, 1]$. We used scenarios $m_0 = 1.0$ and $m_1 = 1.5$ in the rCLBF-QP controller.

For this example, $V$ is parameterized as a two-layer fully-connected neural network with hidden layer size of 48 and $\mathrm{tanh}$ activation. $\pi_\mathrm{NN}$ is represented as a three-layer fully connected network with the same hidden layer size. Training data were sampled from $p_x, p_y, p_z \in [-4, 4]$, $\phi, \theta, \psi \in [-\pi/2, \pi/2]$, and $v_x, v_y, v_z \in [-8, 8]$. We used $\pi_\mathrm{nominal}$ based on an LQR approximation (ignoring state constraints). We set $c=10$, $\lambda=1$, and did not allow relaxations of the constraints in \eqref{eq:rclbf_qp}.

In addition to simulating the performance of our controller, we can also examine the learned rCLBF function itself. A contour plot of $V$ as a function of $p_x$ and $p_z$ (all other states set to zero) is shown in Fig.~\ref{fig:quad9d_V}, comparing the level set at $V(x)=0$ to the safe/unsafe boundaries. In computing this plot, we also computed the maximum violation of the rCLBF condition \eqref{eq:rclbf_decrease_cond}, which we found to be $1.9\times10^{-4}$. This plot was computed by sampling uniformly from $p_x$ and $p_z$, and the maximum distance between adjacent grid points was $0.008$. Based on this extremely small maximum violation (which occurs in a relatively small region of the space), we can conclude that our learning approach yields an rCLBF that is valid throughout most of the domain $\cX$. The violation regions shown in this plot differ from those highlighted in Fig.~\ref{fig:uav_stabilization} because the grid size in Fig.~\ref{fig:quad9d_V} is much smaller, highlighting the sparsity of violations in the state space. An interesting area for future work might involve counter-example guided training of the rCLBF, as in \cite{Chang2019}, to resolve these sparse violations. On the other hand, the theory of almost Lyapunov functions \cite{liu2020almost} suggests that these sparse, low-magnitude violations may not necessarily invalidate the learned CLBF.

\begin{figure}[h]
    \centering
    \includegraphics[width=0.8\linewidth]{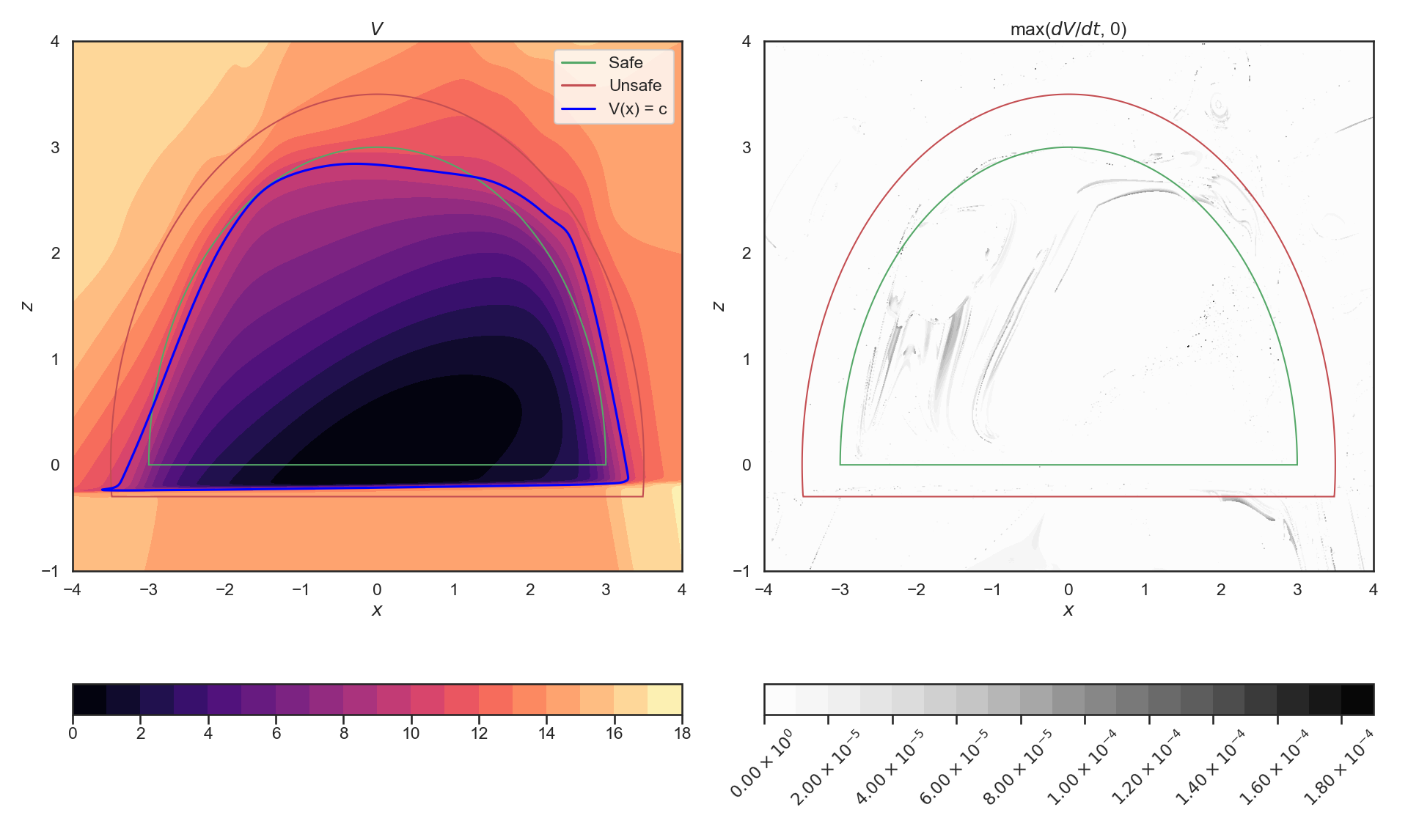}
    \caption{A contour plot of the learned rCLBF $V$ (right) and violation of condition \eqref{eq:rclbf_decrease_cond} (left) for the 3D quadrotor hovering task. The violation of the rCLBF decrease condition \eqref{eq:rclbf_decrease_cond}, which was found to be at most $1.9\times10^{-4}$ over this range, was computed as $\max(dV/dt, 0)$, summed over both parameter scenarios.}
    \label{fig:quad9d_V}
\end{figure}

\subsection*{Neural Lander}

The state of the neural lander is given by $x = [p_x, p_y, p_z, v_x, v_y, v_z]$, with control vector $u = [f_x, f_y, f_z]$. $p_z$ is defined to be positive upwards in this case. This model was developed in \cite{liu2020robust}. The system is parameterized by mass $m$. The dynamics are given by $\dot{x} = f(x) + g(x) u$, with

\begin{align}
    f(x) &= \mat{v_x, v_y, v_z, F_{a1}/m, F_{a2}/m, F_{a3}/m - g}^T \\
    g(x) &= \mat{
        0 & 0 & 0 \\
        0 & 0 & 0 \\
        0 & 0 & 0 \\
        1/m & 0 & 0\\
        0 & 1/m & 0 \\
        0 & 0 & 1/m
    }
\end{align}
where $g$ is the gravitational acceleration and $F_{a}$ is the learned disturbance due to ground effect, represented as a 4-layer neural network. The presence of a learned component in the dynamics means that many traditional control synthesis techniques (including sum-of-squares techniques) do not apply to this system. As with the 9-dimensional quadrotor, these dynamics are not affine in $m$, but they are affine in $1/m$, which can be treated as the uncertain parameter without loss in generality.

We use a similar safe hover task to that used for the 3D quadrotor, with $\xg$ at the origin, $\xs = \set{x\ :\ p_z \geq -0.05 \land \norm{x} \leq 3}$, and $\xu = \set{x\ :\ p_z \leq -0.3 \lor \norm{x} \geq 3.5}$. The mass of the vehicle is sampled uniformly from $m \in [1.47, 2.00]$ and initial conditions $x(0) = [0.5, 0.5, 0.5, 0.5, 0.5, -1.0]$.

For this example, $V$ is parameterized as a two-layer fully-connected neural network with hidden layer size of 48 and $\mathrm{tanh}$ activation. $\pi_\mathrm{NN}$ is represented as a three-layer fully connected network with the same hidden layer size. Training data were sampled from $p_x, p_y \in [-5, 5]$, $z \in [-0.5, 2]$, and $v_x, v_y, v_z \in [-1, 1]$. We used $\pi_\mathrm{nominal}$ based on an LQR approximation (ignoring state constraints and learned dynamics). For completeness, the learned rCLBF for the neural lander with a high-resolution plot of the violation region is included in Fig.~\ref{fig:nl_V}. We set $c=10$, $\lambda=0.1$, and penalized relaxations of the constraints in \eqref{eq:rclbf_qp} with penalty coefficient $7$.

\begin{figure}[h]
    \centering
    \includegraphics[width=1.0\linewidth]{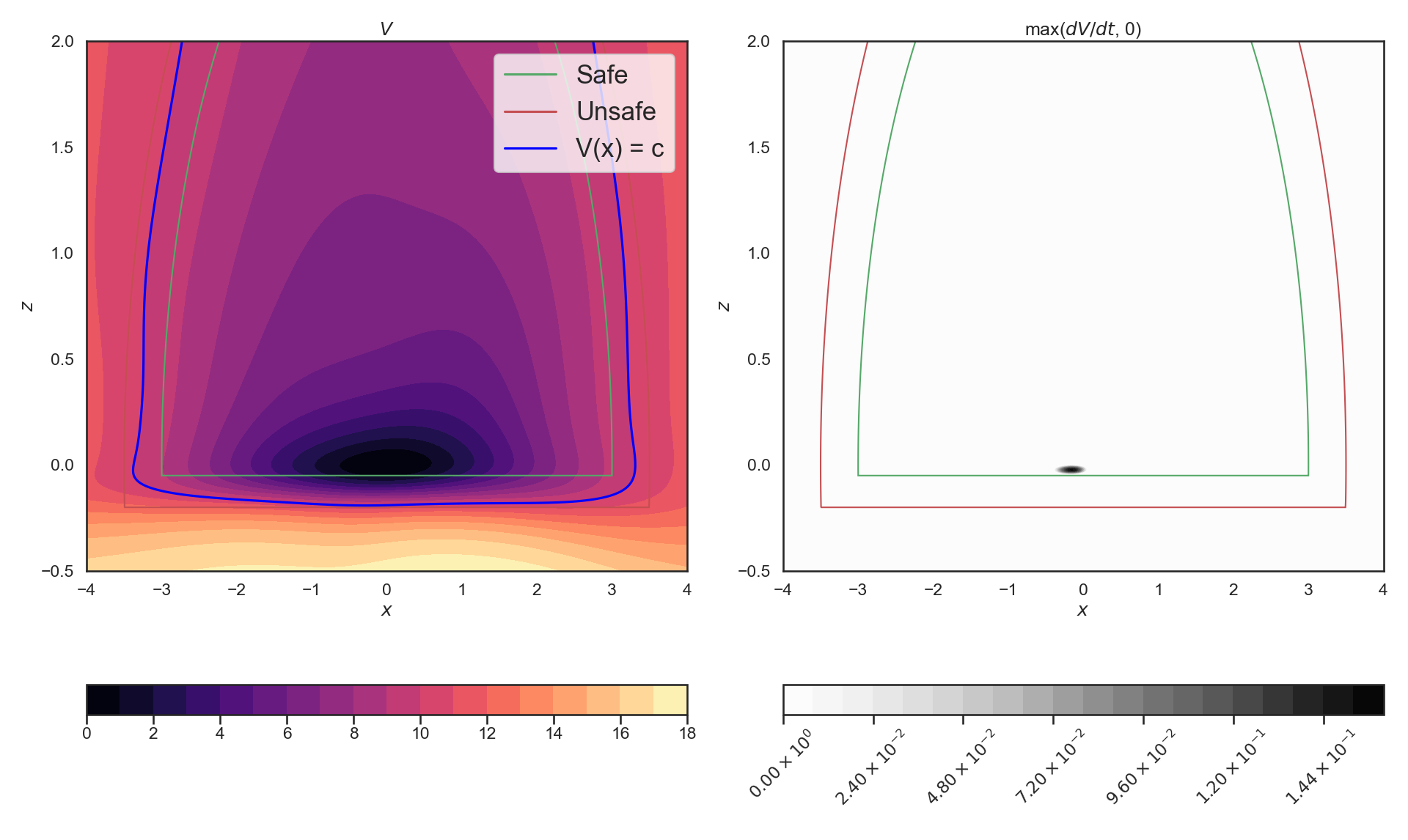}
    \caption{A contour plot of the learned rCLBF $V$ (right) and violation of condition \eqref{eq:rclbf_decrease_cond} (left) for the neural lander hovering task. The violation of the rCLBF decrease condition was computed as $\max(dV/dt, 0)$, summed over both parameter scenarios. We found the violation to be very small and restricted to a small region of the state space. This plot was computed on a grid in $p_x$ and $p_z$ (all other states set to zero). The maximum distance between grid points is $0.008$.}
    \label{fig:nl_V}
\end{figure}

\subsection*{2D Quadrotor}

The state of the 2D quadrotor model is given by $x = [p_x, p_z, \theta, v_x, v_y, \dot{\theta}]$, with control vector $u = [u_1, u_2]$. $p_z$, $u_1$, and $u_2$ are defined to be positive upwards. This model is adapted from \cite{Sun2020}. The system is parameterized by mass $m$, rotational inertia $I$, and the distance of the rotors from the center of mass $r$, and we take $m$ and $I$ to be the uncertain parameters. The dynamics are given by $\dot{x} = f(x) + g(x) u$, with

\begin{align}
    f(x) &= \mat{v_x, v_z, \dot{\theta}, 0, -g, 0}^T \\
    g &= \mat{
        0 & 0 \\
        0 & 0 \\
        0 & 0 \\
        (1/m)\sin\theta & (1/m)\sin\theta \\
        (1/m)\cos\theta & (1/m)\cos\theta \\
        r/I & -r/I
    }
\end{align}
where $g$ is the gravitational acceleration. These dynamics are not affine in $m$ and $I$, but they are affine in $1/m$ and $1/I$, allowing the use of our rCLBF approach.

$\xu$ is set to be the region inside the obstacles, and $\xs$ is offset from the obstacle boundaries by $0.1$ m. To prevent the controller from driving the system out of region covered by the training data, we include a norm constraint in the safe and unsafe sets, $\norm{x} \leq 4.5$ in $\xs$ and $\norm{x} \geq 5$ in $\xu$. To model uncertainty in the mass and inertia of the quadrotor, we vary mass and inertia in $(m, I) \in [1.0, 1.05] \times [0.01, 0.0105]$, with nominal values $m_0 = 1.0$ and $I_0 = 0.01$ (the extreme points of this set are used as scenarios in the rCLBF-QP method).

For this example, $V$ is parameterized as a two-layer fully-connected neural network with hidden layer size of 48 and $\mathrm{tanh}$ activation. $\pi_\mathrm{NN}$ is represented as a three-layer fully connected network with the same hidden layer size. Training data were sampled from $p_x, p_z \in [-4, 4]$, $\theta \in [-\pi, \pi]$, $v_x, v_z \in [-10, 10]$, and $\dot{\theta} = [-2\pi, 2\pi]$. We used $\pi_\mathrm{nominal}$ based on an LQR approximation (ignoring obstacles). We set $c=1$, $\lambda=6$, and penalized relaxations of the constraints in \eqref{eq:rclbf_qp} with penalty coefficient $1100$.

To gain insight into the performance of our learning-based approach to rCLBF synthesis, we can examine the contour plot of the learned rCLBF $V(x)$, shown in Fig.~\ref{fig:2d_quad_obstacles_V}.  These contours were computed on a grid in $p_x$ and $p_z$, with other states set to zero and the maximum distance between adjacent grid points equal to $0.003$. We see that the level set of the learned rCLBF at $V(x) = c$ aligns well with the boundaries of the obstacles, and that the rCLBF derivative condition is satisfied in most of the state space (a small violation $\leq 7.3\times10^{-2}$ is observed in a small region).

\begin{figure}[h]
    \centering
    \includegraphics[width=0.8\linewidth]{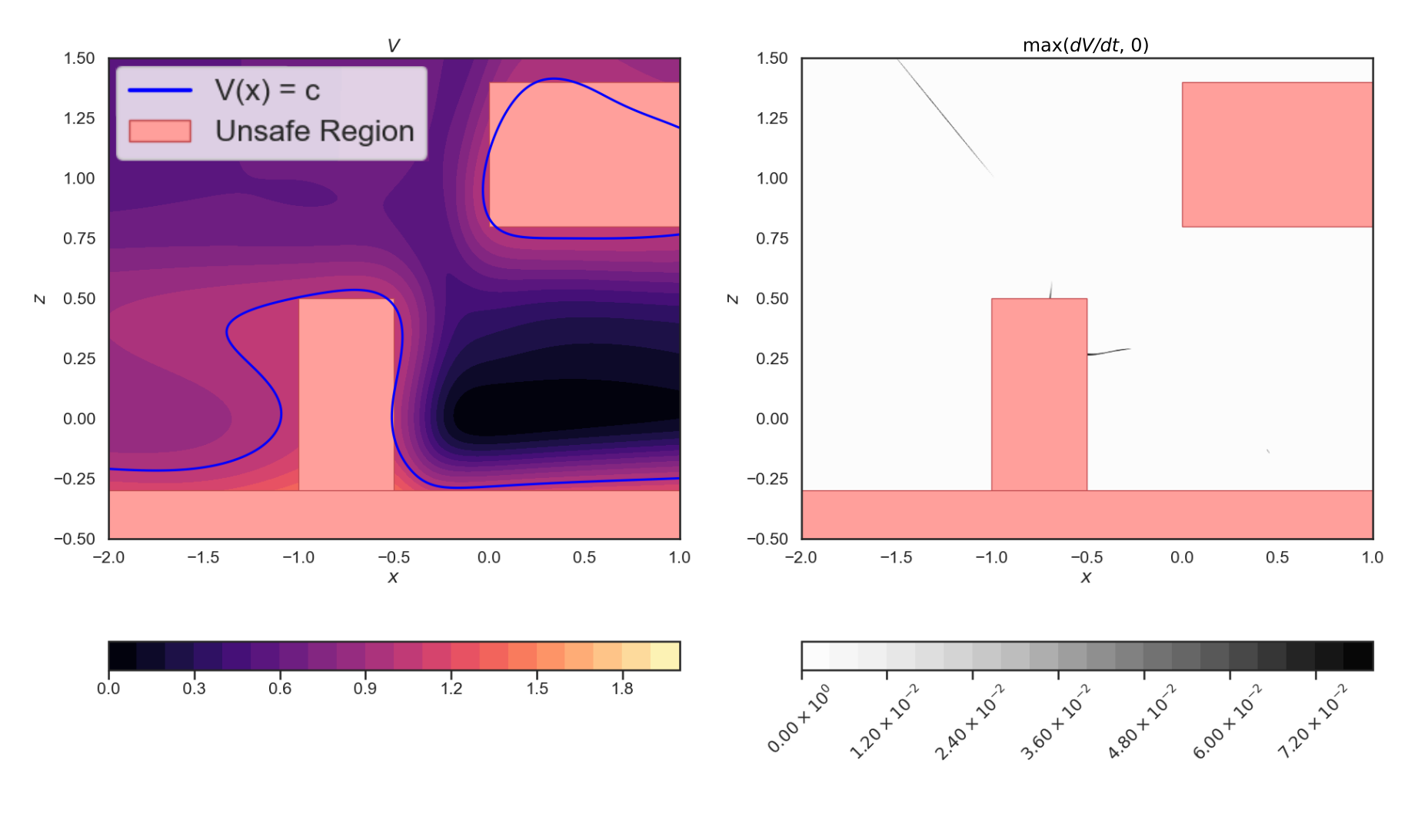}
    \caption{A contour plot of the learned rCLBF $V$ (left) and violation of condition \eqref{eq:rclbf_decrease_cond} (right) for the 2D quadrotor with obstacles. The violation of the rCLBF decrease condition \eqref{eq:rclbf_decrease_cond}, which was found to be at most $7.3\times 10^{-2}$ over this range, was computed as $\max(dV/dt, 0)$, summed over all parameter scenarios.}
    \label{fig:2d_quad_obstacles_V}
\end{figure}

\subsection*{Satellite}

In this example, we consider the satellite rendezvous and docking task adopted from \cite{jewison2016spacecraft}. As is shown in Fig.~\ref{fig:headline}, the blue chaser satellite attempts to close the distance to the black target satellite. Within the green dashed circle, the chaser must stay in the green sector, which represents the line-of-sight (LOS) region where the satellite's sensors are most effective. While both the target and chaser satellites orbit around the Earth, we choose a relative coordinate system centered at the target. The motion of the chaser with respect to the target can be modeled by the Clohessy-Wiltshire-Hill (CWH) equations~\cite{clohessy1960terminal}. The state $x=[p_x, p_y, v_x, v_y]$ is consisted of the relative position and velocity. The control inputs $u=[f_x, f_y]$ are the forces applied to the chaser satellite. To keep the chaser satellite within the LOS region, we define the safe set as $\mathcal{X}_{\mathrm{safe}} = \{x: 4\leq \sqrt{p_x^2 + p_y^2} \leq 8 \lor (p_y \leq -|p_x| \land \sqrt{p_x^2 + p_y^2} \leq 4) \}$, which represents the green sector plus the ring between the grey circle and the green circle. The unsafe set is defined as $\mathcal{X}_{\mathrm{unsafe}} = \{x : \sqrt{p_x^2 + p_y^2} \geq 9 \lor (p_y \geq -|p_x| \land \sqrt{p_x^2 + p_y^2} \leq 3)\}$.

There are two crucial parameters in the model dynamics: the mass of the chaser satellite and the mean-motion $\sqrt{\frac{\mu}{a^3}}$ of the target satellite. $\mu$ is the Earth's gravity constant and $a$ is the length of the semi-major axis of the target's orbit.

The dynamics are given by $\dot{x} = f(x) + g(x)u$, with
\begin{align}
    f(x) &= \mat{v_x, v_y, 2nv_y + 3n^2p_x, -2nv_x}^T \\
    g(x) &= \mat{
        0 & 0 \\
        0 & 0 \\
        1/m & 0 \\
        0 & 1/m
    }
\end{align}

We used a three-layer fully connected neural network with hidden size 256 and $\mathrm{tanh}$ activation to represent the rCLBF $V$. During training, the state samples are uniformly drawn from the state space with range $[-12, 12]$ for each dimension. We set the nominal controller $\pi_{\mathrm{nominal}} = 0$. In the implementation of MPC, we minimize the distance between the goal position and the last step of the planning horizon, subject to the control input constraint $f_x, f_y \in [-20, 20]$. We set constraints to enforce the chaser satellite to enter the green sector rather than anywhere else in the green circle. The planning horizon was 10 steps with timestep 0.02s.

We can examine the learned rCLBF $V$ by plotting its contour in Fig.~\ref{fig:contour_satellite} (left) as a function of $p_x$ and $p_y$ (all other states set to zero). The contour plot shows that the learned $V$ is able to distinguish the safe and unsafe sets. In Fig.~\ref{fig:contour_satellite} (right), we plot the violation of rCLBF decrease condition \eqref{eq:rclbf_decrease_cond}. For most of the samples within the range, the condition is satisfied, and we observe only a slight violation less than $5.5\times 10^{-2}$ for $(p_x, p_y) \in [6, 10]\times[-12, -10]$.

\begin{figure}
    \centering
    \includegraphics[width=\linewidth]{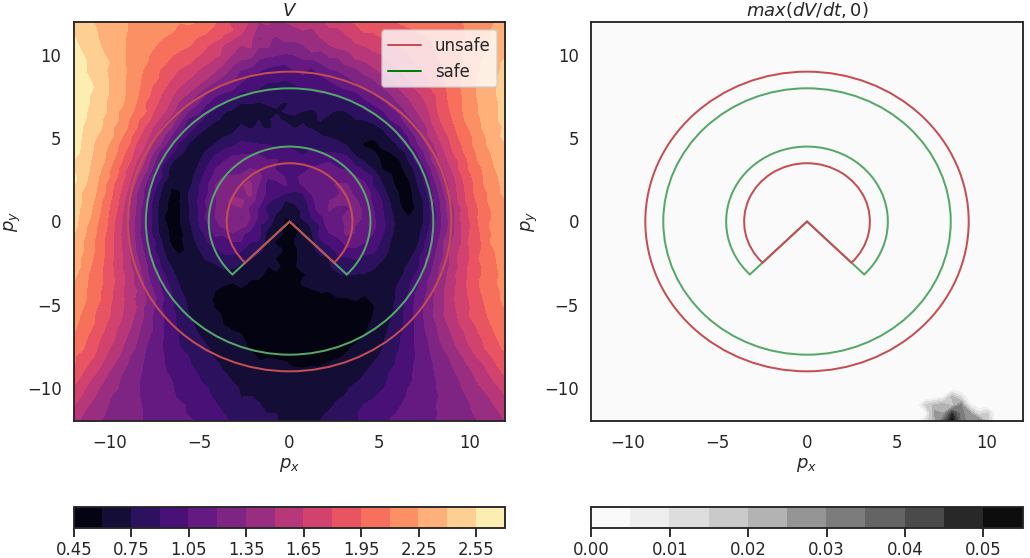}
    \caption{A contour plot of the learned rCLBF $V$ (left) and violation of condition \eqref{eq:rclbf_decrease_cond} (right) for the satellite rendezvous task.}
    \label{fig:contour_satellite}
\end{figure}

\subsection*{Segway}

We consider the Segway obstacle avoidance task illustrated in Fig.~\ref{fig:headline}. The Segway attempts to avoid the obstacle while moving forward, which requires it to tilt forward to avoid collision. The state $x=[p, \theta, v, \omega]$ includes the horizontal position, angle, velocity and angular velocity of the Segway. The control $u$ is the force applied at the base of the system. We assume the vertical position of the wheel's center is always 0 and the length of Segway is 1. The obstacle is a circle with radius 0.1 centered at $(0, 1)$. Denote the position of the Segway's top as $(p_x, p_y) = (p+\sin(\theta), \cos(\theta))$. Then the unsafe set $\mathcal{X}_{\mathrm{unsafe}} = \{x| \sqrt{p_x^2+(p_y-1)^2} \leq 0.1\}$. We define the safe set as $\mathcal{X}_{\mathrm{safe}} = \{x| \sqrt{p_x^2+(p_y-1)^2} \geq 0.15\}$.

The Segway model is from Chapter 3.2 of \cite{aastrom2021feedback}. The state $x = [p, \theta, v, \omega]$ with control input $u$ as the force aligned with $p$ applied at the base of the system. Let $M$ be the mass of the base, $m$ and $J$ be the mass and inertia of the system to be balanced. Denote the distance from the base to the center of mass of the system as $l$. Let $g$ be the gravity constant. Define $M_t = M+m$ as the total mass and $J_t = J+ml^2$ be the total inertia. The system dynamics are given by $\dot{x} = f(x) + g(x)u$, with
\begin{align}
    f(x) &= \mat{
    v\\
    \omega\\
    \frac{gs_{\theta}c_{\theta} + \lambda_1vc_{\theta}+\lambda_2v-l\omega^2s_{\theta}}{c_{\theta} - \frac{M_tJ_t}{m^2l^2} + \lambda_9} \\
    \frac{\lambda_3vc_{\theta}+\lambda_4v-\frac{M_tg}{ml}s_{\theta}-\omega^2s_{\theta}c_\theta}{c_{\theta}^2-\frac{M_tJ_t}{m^2l^2}+ \lambda_9}
    }\\
    g(x) &= \mat{
    0\\
    0\\
    \frac{\frac{\lambda_6}{M_t}(\lambda_5+c_{\theta})}{c_{\theta}^2-\frac{M_tJ_t}{m^2l^2}+ \lambda_9}\\
    \frac{\frac{\lambda_8l}{J_t}(c_{\theta}+\lambda_7)}{c_{\theta}^2-\frac{M_tJ_t}{m^2l^2}+ \lambda_9}
    }
\end{align}
In the implementation of the rCLBF $V$, we used a three-layer fully connected neural network with hidden size 64 and $\mathrm{tanh}$ activation. The lower and upper bound of the state space are $[-3, -\pi/2, -1, -3]^T$ and $[3, \pi/2, 1, 3]^T$. Training samples are uniformly sampled from this range. We used an LQR controller for $\pi_{\mathrm{nominal}}$. In the MPC baseline, we minimize the angle w.r.t. the vertical axis at each step, subject to the top of the Segway being outside the unsafe set. At each step, when the solver fails find a feasible solution satisfying the constraint, we use $u=0$ for that step. The planning horizon is 10 steps, and the timestep is set to 0.005. 

The contour plot of the learned $V$ and the violation of the rCLBF condition are shown in Fig.~\ref{fig:contour_Segway}. We set $v$ and $\omega$ to zero, then sample $p$ and $\theta$ to evaluate $V$. To make this plot more interpretable, we convert $p$ and $\theta$ to the $x-y$ coordinates of the top of the Segway. Fig.~\ref{fig:contour_Segway} shows that the learned $V$ has a wider safe and unsafe set than $\xs$ and $\xu$, which explains why the simulated rCLBF-QP trajectories give such a wide berth to the obstacle. The Segway learns to gradually tilt down when it is 0.5m away from the obstacle, instead of abruptly changing its angle when it is too close to the obstacle. In addition, in Fig.~\ref{fig:contour_Segway} (right) we observe only a minor violation of the rCLBF decrease condition; we have $dV/dt \leq 0$ for most of the plotted range.

\begin{figure}
    \centering
    \includegraphics[width=\linewidth]{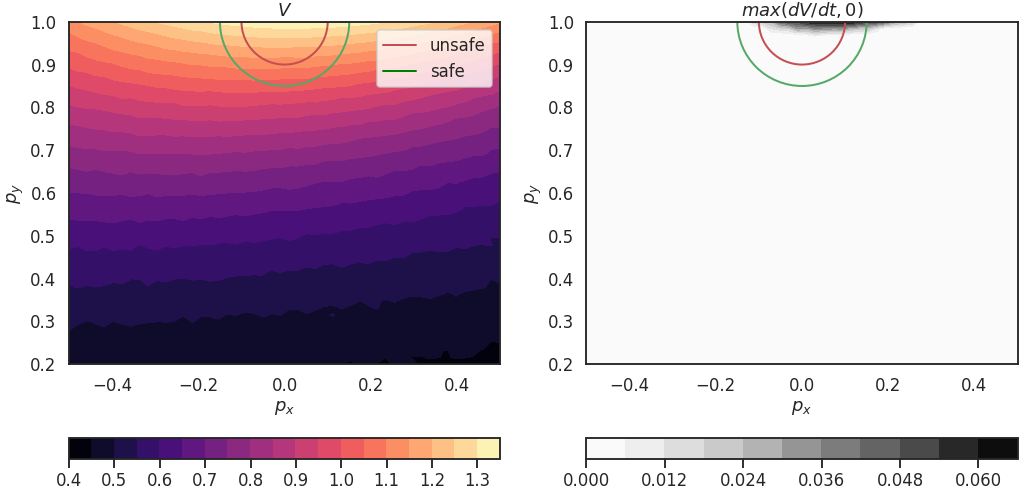}
    \caption{A contour plot of the learned rCLBF $V$ (left) and violation of condition \eqref{eq:rclbf_decrease_cond} (right) for the Segway obstacle avoidance task.}
    \label{fig:contour_Segway}
\end{figure}

\subsubsection*{Failure of Constrained Policy Optimization}

We attempted to train a controller for this task using the constrained policy optimization reinforcement learning algorithm (CPO,~\cite{cpo}). We found that the RL agent was able to stabilize the Segway in the absence of obstacles, but it failed to stabilize the system when an obstacle was included during training. An example plot of the trained controller's performance is shown in Fig.~\ref{fig:cpo_segway}.

\begin{figure}
    \centering
    \includegraphics[width=0.4\linewidth]{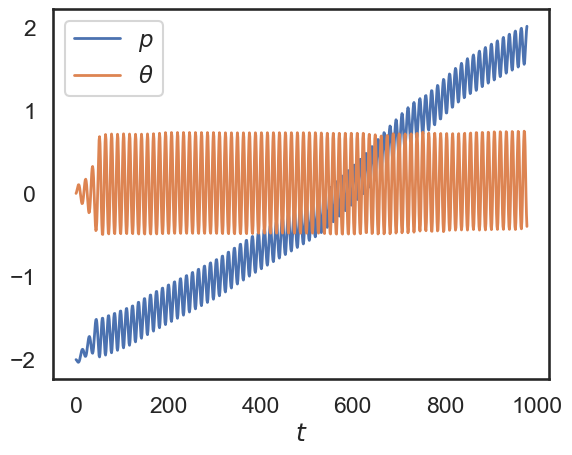}
    \includegraphics[width=0.4\linewidth]{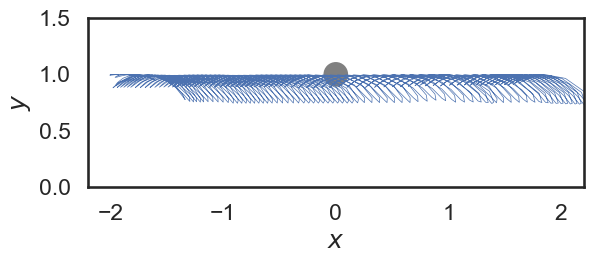}
    \caption{A plot of the CPO RL agent's performance on the Segway obstacle-avoidance task. Left: $p$ and $\theta$ over time; Right: the path of the head of the Segway in the $xy$ plane, showing collision between the Segway and the obstacle.}
    \label{fig:cpo_segway}
\end{figure}

\newpage

\section*{Pandemic Considerations}

Our paper is concerned with synthesizing controllers for robotic systems. Due to facility access limitations from the COVID-19, we were not able to gather experimental results on hardware, so our paper focuses on experiments conducted in simulation. We took a number of steps to ensure that performance in our simulations correlates with expected performance in hardware. In particular,
\begin{enumerate}
    \item We report evaluation times for all controllers used in our experiments and compare these times to the control frequency, allowing us to determine whether the algorithms could feasibly be deployed in real-time.
    \item We randomly vary the values of model parameters while computing safety and error rates, simulating the uncertainty present in models of physical systems.
    \item Our framework can be easily extended to include physical constraints, particularly actuator limits, within the QP-based controller.
    \item One of our simulated examples (the neural lander) includes a learned model of aerodynamic ground effect, which uses experimental data to make the simulation more realistic by including otherwise unmodeled effects.
\end{enumerate}

That said, there are a number of gaps between our simulation and reality. The most glaring gaps are that
\begin{enumerate}
    \item Although our framework supports physical constraints, we do not present a thorough evaluation of the effect of varying actuator limits on controller performance.
    \item We assume full information about the robot state, which in practice means that we assume a high-quality state estimate is available at a suitably high frequency.
    \item We do not study the effects of delay on the stability or safety of our controller (beyond our measurement of control frequency).
\end{enumerate}

In the coming months, we hope to carry out hardware demonstrations that justify these assumptions and close these gaps.

\end{document}